\documentclass[letterpaper, 10 pt, journal, doublecolumn]{IEEEtran}

\IEEEoverridecommandlockouts                              

\usepackage{graphicx} 
\usepackage{amsmath,bm,times} 
\usepackage{amssymb}  
\usepackage{bm}
\usepackage{tikz}
\usetikzlibrary{shapes,arrows,backgrounds,fit,positioning}
\usepackage{subfigure}
\usepackage{cite}
\usepackage{booktabs} 
\usepackage[linesnumbered, ruled, vlined]{algorithm2e}
\usepackage{tabularx}
\usepackage{url}
\usepackage{balance}

\newtheorem{assumption}{Assumption}
\newtheorem{problem}{Problem}
\newtheorem{proposition}{\bf Proposition}
\newtheorem{remark}{Remark}
\newtheorem{theorem}{\bf Theorem}
\newtheorem{definition}{Definition}
\newtheorem{corollary}{\bf Corollary}
\newtheorem{lemma}{\bf Lemma}
\newtheorem{example}{Example}

\newcommand{\QED}{\hfill$\blacksquare$}

\newcommand{\mycomment}[1]{
}%
\newcommand{\change}[1]{{\color{black} #1}}

\DeclareMathOperator*{\argmax}{arg\,max}

\title{\LARGE \bf Learning with Delayed Payoffs in Population Games\\ using Kullback-Leibler Divergence Regularization}

\author{Shinkyu Park and Naomi Ehrich Leonard
  \thanks{Park's work was supported by funding from King Abdullah University of Science and Technology (KAUST). Leonard's work was supported in part by ONR grant N00014-19-1-2556 and ARO grant W911NF-18-1-0325.}
  \thanks{Park is with Electrical and Computer Engineering, King Abdullah University of Science and Technology (KAUST), Thuwal 23955, Saudi Arabia. {\tt shinkyu.park@kaust.edu.sa}}
  \thanks{Leonard is with the Department of Mechanical and Aerospace Engineering, Princeton University, Princeton, NJ 08544, USA. {\tt naomi@princeton.edu}}
}

\IEEEoverridecommandlockouts
\begin{document}

\maketitle

\begin{abstract}
  We study a multi-agent decision problem in large population games. Agents from multiple populations select strategies for repeated interactions with one another. At each stage of these interactions, agents use their decision-making model to revise their strategy selections based on payoffs determined by an underlying game. Their goal is to learn the strategies that correspond to the Nash equilibrium of the game. However, when games are subject to time delays, conventional decision-making models from the population game literature may result in oscillations in the strategy revision process or convergence to an equilibrium other than the Nash.
  To address this problem, we propose the Kullback-Leibler Divergence Regularized Learning (KLD-RL) model, along with an algorithm that iteratively updates the model's regularization parameter across a network of communicating agents. Using passivity-based convergence analysis techniques, we show that the KLD-RL model achieves convergence to the Nash equilibrium without oscillations, even for a class of population games that are subject to time delays. We demonstrate our main results numerically on a two-population congestion game and a two-population zero-sum game.

\end{abstract}



\section{Introduction} \label{sec:intro}
Consider a large number of agents engaged in repeated strategic interactions. Each agent selects a strategy for these interactions and can repeatedly revise its strategy according to payoffs determined by an underlying payoff mechanism. To learn and adopt the best strategy without knowing the structure of the payoff mechanism, an agent needs to revise its strategy selection at each stage of the interactions based on the instantaneous payoffs it receives.
  This multi-agent decision problem is relevant in control engineering applications where the goal is to design decision-making models that enable multiple agents to learn effective strategies in a self-organized way. As discussed in \cite[\S III]{Park2019Payoff-Dynamic-}, applications include multi-agent task allocation, demand response in smart grids, and wireless networks, as well as distributed control systems and optimization, among others.

To formalize the problem, we adopt the \textit{population game} framework \cite[\S 2]{Sandholm2010-SANPGA-2}. In this framework, a \textit{payoff function} defines how the payoffs are determined based on the agents' strategy profile, which is the distribution of their strategy selections over a finite number of available strategies. \textit{An evolutionary dynamics model} describes how individual agents revise their strategies to increase the payoffs they receive. A key research theme involves establishing convergence of the strategy profile to the Nash equilibrium, where no agent is better off by unilaterally revising its strategy selection.\footnote{In this work, we consider that the Nash equilibrium represents a desired distribution of the agents' strategy selections, e.g., the distribution of route selections minimizing road congestion in congestion games (Example~\ref{example:congestion_game}) or minimizing opponents' maximum gain in zero-sum games (Example~\ref{example:zero_sum_game}), and investigate convergence of the agents' strategy profile to the Nash equilibrium. However, we note that such an equilibrium is not always desirable and may result in the worst outcome, for instance, in social dilemmas as illustrated by the prisoner's dilemma \cite{doi:10.1146/annurev.soc.24.1.183}. We refer to \cite{9663230} and references therein for other studies on decision model design in social dilemmas.}

  Unlike in existing studies, we investigate scenarios where the payoff mechanism is subject to time delays. This models, for example, the propagation of traffic congestion on roads in congestion games, delays in communication between electric power utilities and demand response agents in demand response games, and limitations of agents in sensing link status in network games. When agents revise their strategy selections based on delayed payoffs, the strategy profile does not converge to the Nash equilibrium, as discussed in prior work in the game theory literature \cite{ALBOSZTA2004175, oaku2002, bodnar2020, Khalifa2018, IIJIMA20121, YI1997111, tembine2011, wesson2016, PhysRevE.101.042410, WANG20178, 7039982, Obando2016, doi:10.1142/S0218127413501228, 10.5555/1345263.1345309, HU2019218, RePEc:eee:matsoc:v:61:y:2011:i:2:p:83-85}.

As a main contribution of this paper, we propose a new class of decision-making models called \textit{Kullback-Leibler Divergence Regularized Learning (KLD-RL)}. Agents using this model revise their strategies based on received payoffs, augmented with a regularization term defined by the KL divergence. As detailed in \S\ref{sec:regularizedLogit}, this regularization prevents large deviations of agents' strategy profile from the one specified by the model's regularization parameter. Consequently, with the proper selection of this parameter, using the KL divergence for regularization renders agents' strategy revisions insensitive to time delays in the payoff mechanism. This prevents oscillation in the agents' strategy revision process and, through successive updates of the regularization parameter, ensures that agents improve their strategy selections. As a result, their strategy profile is guaranteed to converge to the Nash equilibrium in a certain class of population games.

The logit dynamics model \cite{10.2307/3081987, HOFBAUER200747} is known to converge to an equilibrium state in a large class of population games, including games subject to time delays \cite{Park2019Payoff-Dynamic-, 9029756}.
However, as discussed in \cite{HOFBAUER200747}, the equilibrium state of the logit dynamics model is a perturbed version of the Nash equilibrium. This forces agents to select sub-optimal strategies, for instance, in potential games \cite{MONDERER1996124, SANDHOLM200181} with concave payoff potentials, where the Nash equilibrium is the socially optimal strategy profile. Such a significant limitation in existing models motivates our investigation of a new decision-making model.

  Below, we summarize the main contributions of this paper:
  \begin{itemize}
  \item We propose a parameterized class of KLD-RL models that generalize the existing logit dynamics model. We explain how the new class of models implements the idea of regularization in multi-agent decision-making and provides an algorithm that iteratively updates the model's regularization parameter.

  \item Leveraging stability results from recent works on higher-order learning in large population games
    \cite{9029756, Park2019Payoff-Dynamic-}, we discuss the convergence of the strategy profile to the Nash equilibrium under the KLD-RL model in an important class of population games, widely known as \textit{contractive population games} \cite{SANDHOLM2015703}. We also introduce a distributed algorithm for updating the regularization parameter, requiring only communication between neighboring agents, as defined by an underlying communication graph.

  \item We present simulations using multi-population games to demonstrate how the new model ensures convergence to the Nash equilibrium, despite time delays in the games.
  \end{itemize}

  The paper is organized as follows. In \S\ref{sec:problem_description}, we explain the multi-agent decision problem addressed in this paper. In \S\ref{sec:literature_review}, we provide a comparative review of related works. In \S\ref{sec:regularizedLogit}, we introduce the KLD-RL model and explain how to iteratively update the model's regularization parameter. Our main results establish the convergence of the strategy profile, determined by the model, to the Nash equilibrium in a class of contractive population games. In \S\ref{section:distributed_parameter_update}, we discuss a distributed approach to update the regularization parameter. In \S\ref{sec:simulation}, we present simulation results that demonstrate the effectiveness of the proposed model in learning and converging to the Nash equilibrium. We conclude the paper with a summary and future directions in \S\ref{sec:conclusion}.

  Table~\ref{tab:notation} summarizes the basic notation used throughout the paper. For all variables and parameters adopted in this paper, superscripts are used to indicate their association with a specified population, unless otherwise noted.

  \section{Problem Description} \label{sec:problem_description}
  \begin{table}
    \change{
      \begin{tabularx} {.48\textwidth} {l | X}

    $\mathbb{R}^n, \mathbb{R}^n_{\geq 0}$ &  sets of  $n$-dimensional real vectors and nonnegative real vectors.
    \\
    
    $\mathbb{X}^k, \mathbb X$ &  state spaces of population~$k$ and the society. \\

    $T\mathbb{X}^k, T\mathbb{X}$ &  tangent spaces of $\mathbb X^k$ and $\mathbb X$. \\

    $\mathrm{int}(\mathbb{X}^k), \mathrm{int}(\mathbb{X})$ &  interiors of $\mathbb X^k$ and $\mathbb X$ defined as \\ & 
    $\mathrm{int}(\mathbb{X}^k) = \{ x^k \in \mathbb{X}^k \,|\, x_i^k >0, \, 1 \leq i \leq n^k\}$ and \\ & $\mathrm{int}(\mathbb{X}) = \mathrm{int}(\mathbb X^1) \times \cdots \times \mathrm{int}(\mathbb X^M)$. \\

    $\mathcal F, D \mathcal F$ &  payoff function and its differential. \\

    $\mathbb{NE} \left( \mathcal F \right)$ &  Nash equilibrium set of $\mathcal F$, defined in \eqref{eq:nash_equilibrium}. \\
    
    $\mathbb{PNE}_{\eta, \theta} \left( \mathcal F \right)$ &  perturbed Nash equilibrium set of $\mathcal F$, defined in \eqref{eq:perturbed_nash}. \\
    
    $\mathcal D (x \,\|\, y)$ &  Kullback-Leibler divergence defined as $\sum_{i=1}^{n^k} x_i \ln \frac{x_i}{y_i}$ for $x,y \in \mathrm{int} (\mathbb X^k)$.

  \end{tabularx}
  \caption{List of basic notation}
  \label{tab:notation}
  \vspace{-5ex}
  }
\end{table}

Consider a \textit{society} consisting of $M$ \textit{populations} of decision-making \textit{agents}.\footnote{We adopt materials on population games and relevant preliminaries from \cite[\S 2]{Sandholm2010-SANPGA-2}.} We denote by $\{1, \cdots, M\}$ the populations constituting the society and by $\{1, \cdots, n^k \}$ the set of strategies available to agents in each population~$k$. Let $x^k(t) = (x_1^k(t), \cdots, x_{n^k}^k(t))$ be an $n^k$-dimensional nonnegative real-valued vector, where each entry $x_i^k(t)$ denotes the portion of population~$k$ adopting strategy~$i$ at time instant~$t$. We refer to $x^k(t)$ as the \textit{state} of population~$k$ and to the constant $m^k = \textstyle\sum_{i=1}^{n^k} x_i^k(t), ~ \forall t \geq 0$, as the \textit{mass} of the population. Also, by aggregating the states of all populations, we define the \textit{social state} $x(t) = (x^1(t), \cdots, x^M(t))$, which describes the strategy profiles across all $M$ populations at time~$t$. Let $n$ be the total number of strategies available in the society, i.e., $n = \sum_{k=1}^M n^k$. We denote the space of viable population states as $\mathbb X^k = \{ x^k \in \mathbb R_{\geq 0}^{n^k} \,\big|\, \textstyle\sum_{i=1}^{n^k} x_i^k = m^k \}$. Accordingly, we define $\mathbb X = \mathbb X^1 \times \cdots \times \mathbb X^M$ as the space of viable social states. For concise presentation, without loss of generality, we assume that $m^k = 1, \forall k \in \{1, \cdots, M \}$.
  We denote the payoff given to the agents in population~$k$ who select strategy~$i$ at time~$t$ as $p_i^k(t)$. We define $p^k(t) = (p_1^k(t), \cdots, p_{n^k}^k(t)) \in \mathbb R^{n^k}$ to represent a vector of payoffs assigned to each population~$k$, and $p(t) = (p^1(t), \cdots, p^M(t)) \in \mathbb R^n$ to refer to a vector of payoffs assigned to all agents across the entire society.

We formulate our main problem as designing an agent decision-making model that determines the social state $x(t)$ in response to the payoff vector $p(t)$, and establishing the convergence of $x(t)$ to the Nash equilibrium of the underlying game, especially when the mechanism that determines $p(t)$ -- which depends on $x(t)$ -- is subject to time delays. To solve this problem, using a formalism originally introduced in \cite{Fox2013Population-Game, Park2019Payoff-Dynamic-, 9029756}, we express the agent decision-making model as the so-called \textit{evolutionary dynamics model (EDM)} and the payoff mechanism with time delays as the so-called \textit{payoff dynamics model (PDM)}. We then study the stability of the feedback interconnection of the two models. As discussed in these references, such a compositional approach allows us to adopt the notion of passivity for the study, which we detail in \S\ref{sec:regularizedLogit}. In what follows, we review relevant definitions from the literature.

\subsection{Population Games and Time Delays in Payoff Mechanisms}
\subsubsection{Population games} \label{sec:population_game}

The game is associated with a \textit{payoff function} $\mathcal{F} = (\mathcal{F}^1, \cdots, \mathcal{F}^M)$, where $\mathcal{F}^k: \mathbb{X} \to \mathbb{R}^{n^k}$ assigns a payoff vector to each population~$k$ as $p^k(t) = \mathcal{F}^k(x(t))$. We adopt the following definition of the Nash equilibrium for $\mathcal{F}$.
\begin{definition}[Nash Equilibrium] \label{def:nash}
  An element $x^{\text{NE}}$ in $\mathbb{X}$ is called the \textit{Nash equilibrium} of the population game $\mathcal{F}$ if it satisfies the condition:
  \begin{align} \label{eq:nash_equilibrium}
    (x^{\text{NE}} - z)' \mathcal{F}(x^{\text{NE}}) \geq 0, ~ \forall z \in \mathbb{X}.
  \end{align}
\end{definition}
Population games can have multiple Nash equilibria. We denote by $\mathbb{NE}(\mathcal{F})$ the set of all Nash equilibria of $\mathcal{F}$. Below, we provide examples of population games and identify their unique Nash equilibria. These examples will be used in \S\ref{sec:simulation} to illustrate our main results.

\begin{figure} [t!]
  \center
  \includegraphics[trim={.0in .0in 0 0}, width=2.2in]{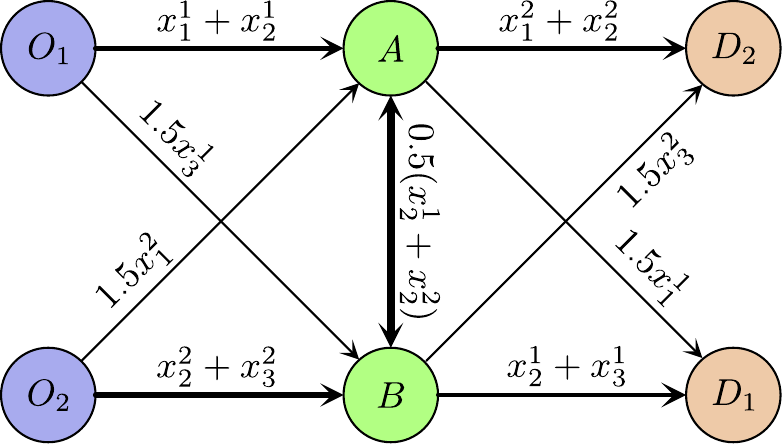}
  \caption{
    Two-Population Congestion Game. Agents in each population $k$ traverse from origin $O_k$ to destination $D_k$ using one of the following routes: $O_1 \to A \to D_1$ (Route~1), $O_1 \to A \to B \to D_1$ (Route~2), and $O_1 \to B \to D_1$ (Route~3) for population~1; $O_2 \to A \to D_2$ (Route~1), $O_2 \to B \to A \to D_2$ (Route~2), and $O_2 \to B \to D_2$ (Route~3) for population~2. We assume that when the same number of agents use the links, the diagonal links ($O_1 \to B, O_2 \to A, A \to D_1, B \to D_2$) are $50 \%$ more congested than the horizontal links, e.g., because the roads represented by the diagonal links are narrower; whereas the vertical link $A \leftrightarrow B$ is $50 \%$ less congested than the horizontal links, e.g., because the road associated with the vertical link is wider. The different weights on the links reflect this assumption. }
  \label{fig:congestion_game}
  \vspace{-1em}
\end{figure}

  \begin{example} 
    \label{example:congestion_game}
    Consider a congestion game with two populations ($M=2$). Each population is assigned a fixed origin and destination. We consider the scenario in which every agent in each population needs to repeatedly travel from its origin to its destination, e.g., to commute to work every workday. To reach their respective destinations, agents in each population use one of three available routes as depicted in Fig.~\ref{fig:congestion_game}. Each strategy in the game is defined as an agent taking one of the available routes. The associated payoff reflects the level of congestion along the selected route, which depends on the number of agents from possibly both populations using the route.

    To formalize this, we adopt the payoff function $\mathcal{F} = (\mathcal{F}^1, \mathcal{F}^2)$ defined as:
    \begin{subequations} \label{eq:congestion_game}
      \begin{align}
        \mathcal{F}^1(x^1, x^2) &= - \begin{pmatrix}
          2.5 x_1^1 + x_2^1 \\
          x_1^1 + 2.5 x_2^1 + x_3^1 + 0.5 x_2^2 \\
          x_2^1 + 2.5 x_3^1
        \end{pmatrix} \\
        \mathcal{F}^2(x^1, x^2) &= - \begin{pmatrix}
          2.5 x_1^2 + x_2^2 \\
          0.5 x_2^1 + x_1^2 + 2.5 x_2^2 + x_3^2 \\
          x_2^2 + 2.5 x_3^2
        \end{pmatrix}.
      \end{align}
    \end{subequations}
    We note that \eqref{eq:congestion_game} has a unique Nash equilibrium $x^{\text{NE}} = \left(4/9, 1/9, 4/9, 4/9, 1/9, 4/9\right)$. Since \eqref{eq:congestion_game} admits a concave potential function, i.e., there exists a concave potential function $f$ satisfying $\mathcal{F} = \nabla f$, the function $f$ attains its minimum at the Nash equilibrium. Consequently, with a particular choice of $f$, defined as the average congestion level across all six routes, $f(x^1, x^2) = (x^1)' \mathcal{F}^1(x^1, x^2) + (x^2)' \mathcal{F}^2(x^1, x^2)$, we can conclude that the average congestion is minimized at the Nash equilibrium. \QED
  \end{example}

\begin{example} 
  \label{example:zero_sum_game}
  Consider a two-population zero-sum game with payoff function $\mathcal{F} = (\mathcal{F}^1, \mathcal{F}^2)$  derived from a \textit{biased} Rock-Paper-Scissors (RPS) game \cite{Omidshafiei2019} as follows:
  \begin{subequations} \label{eq:rps_game}
    \begin{align}
      \mathcal{F}^1(x^1, x^2) &= \begin{pmatrix}
        -0.5 x_2^2 + x_3^2 \\
        0.5 x_1^2 - 0.1 x_3^2 \\
        -x_1^2 + 0.1 x_2^2
      \end{pmatrix} \\
      \mathcal{F}^2(x^1, x^2) &= \begin{pmatrix}
        -0.5 x_2^1 + x_3^1 \\
        0.5 x_1^1 - 0.1 x_3^1 \\
        -x_1^1 + 0.1 x_2^1
      \end{pmatrix}.
    \end{align}
  \end{subequations}
  The study of zero-sum games has important implications in security-related applications. For example, attacker-defender (zero-sum) game formulations \cite{doi:10.1063/1.5029343, Alpcan2015} can be used to predict an attacker’s strategy at the Nash equilibrium and to design the best defense strategy.

  Agents in each population $k \in \{1,2\}$ can select one of three strategies: rock ($x_1^k$), paper ($x_2^k$), or scissors ($x_3^k$). The payoff vector $\mathcal{F}^k(x^1, x^2)$ assigned to each population~$k$ reflects the chances of winning the game against its opponent population, as illustrated in Fig.~\ref{fig:zerosum_game}. The agents are engaged in a multi-round game in which, based on the payoff vector received in the previous round, those who are given the opportunity revise their strategy selections.
  Note that the game has a unique Nash equilibrium $x^{\text{NE}} = (1/16, 10/16, 5/16, 1/16, 10/16, 5/16)$, at which each population minimizes the opponent population's maximum gain (or equivalently maximizes its worst-case (minimum) gain). \QED
\end{example}
\begin{figure} [t!]
  \center
  \includegraphics[trim={.0in .0in 0 0}, width=2.2in]{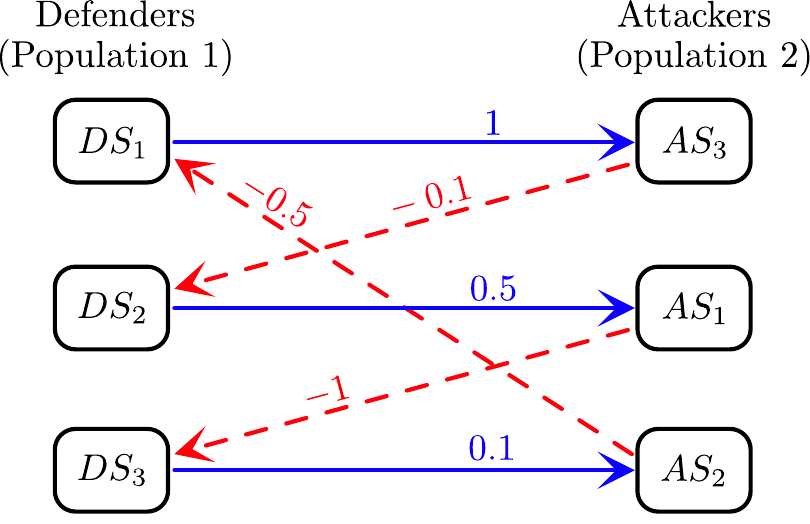}
  \caption{Two-Population Zero-Sum Game. Agents in the defender population (population~1) select defending strategies $(DS_1, DS_2, DS_3)$ to play against those in the attacker population (population~2) who adopt attacking strategies $(AS_1, AS_2, AS_3)$. The positive (negative) weight on the blue (red dotted) arrow between $DS_i$ and $AS_j$ denotes the reward (loss) associated with population~1 when the defenders and attackers adopt $DS_i$ and $AS_j$, respectively. The payoff $\mathcal{F}^1_i(x^1, x^2)$ associated with $DS_i$ of population~1 is the sum of the rewards and losses when $x^2$ is the state of population~2; whereas the payoff $\mathcal{F}^2_j(x^1, x^2)$ associated with $AS_j$ of population~2 is the negative sum of the rewards and losses when $x^1$ is the state of population~1.
  }
  \label{fig:zerosum_game}
  \vspace{-1em}
\end{figure}

We make the following assumption on  payoff function~$\mathcal F$.
  \begin{assumption} 
    \label{assumption:payoff_function}
    The differential $D \mathcal F$ of $\mathcal F$, defined as the map that satisfies $\lim_{h \to 0} \frac{1}{h} (\mathcal F (z + h \tilde z) - \mathcal F (z)) = D \mathcal F(z) \tilde z, ~ \forall z \in \mathbb X, \tilde z \in T\mathbb X$, exists and is continuous on $\mathbb X$.
  \end{assumption}

  Since $\mathbb X$ is compact, both $\mathcal F$ and $D \mathcal F$ are bounded. Consequently, there exist constants $B_{\mathcal F}$ and $B_{D \mathcal F}$ such that
  $B_{\mathcal F} = \max_{z \in \mathbb X} \| \mathcal F (z) \|_2$ and $B_{D \mathcal F} = \max_{z \in \mathbb X} \| D \mathcal F (z) \|_2$, respectively.
Note that any affine payoff function $\mathcal F(x) = Fx + b$, with $F \in \mathbb R^{n \times n}$ and $b \in \mathbb R^n$, e.g., \eqref{eq:congestion_game}, \eqref{eq:rps_game}, satisfies Assumption~\ref{assumption:payoff_function}. Contractive population games are defined as follows.\footnote{Contractive games (strictly contractive games) were previously referred to as stable games (strictly stable games) \cite{RePEc:eee:jetheo:v:144:y:2009:i:4:p:1665-1693.e4}. We adopt the latest naming convention.}
\begin{definition}[Contractive Population Game \cite{SANDHOLM2015703}] \label{def:contractive}
  A population game $\mathcal F$ is called \textit{contractive} if it holds that
  \begin{align} \label{eq:contractive}
    (w - z)' (\mathcal F(w) - \mathcal F(z)) \leq 0, ~ \forall w, z \in \mathbb X.
  \end{align}
\end{definition}
If the equality holds if and only if $w = z$, $\mathcal F$ is called \textit{strictly contractive}.

Both games in Examples~\ref{example:congestion_game} and \ref{example:zero_sum_game} are contractive. If a population game $\mathcal{F}$ is contractive, then its Nash equilibrium set $\mathbb{NE}(\mathcal{F})$ is convex; moreover, if $\mathcal{F}$ is strictly contractive, then it has a unique Nash equilibrium \cite[Theorem 13.9]{SANDHOLM2015703}.

\subsubsection{Time delays in payoff mechanisms}

\change{
  Unlike the standard formulation discussed in \S\ref{sec:population_game}, this work considers a payoff vector that depends not only on the current social state but also on past social states. This approach captures the time delays inherent in the payoff mechanisms of population games. We start by discussing two special cases and then formally define the class of population games with time delays.
}

\paragraph{Payoff Function with a Time-Dependent Delay}
\label{section:delayed_payoff_function}
  Consider that the payoff vector $p(t)$ at time $t$ depends on the past social state $x(t-d(t))$, where $d(t)$ is a non-negative differentiable function denoting a time-dependent delay: 
  \begin{align} \label{eq:population_game_time_delay}
    p(t) = \mathcal F ( x(t-d(t))), ~t \geq \tau_0 \geq 0
  \end{align}
with $p(t) = \mathcal F(x(0)), ~ \tau_0 > t \geq 0$. The time constant $\tau_0$ is chosen to ensure that $t - d(t) \geq 0$ for all $t \geq \tau_0$. We assume fixed positive constants $B_{d}$ and $B_{\dot d}$ such that $d(t) \in [0, B_{d}]$ and $\dot d(t) \in [-B_{\dot d}, 1)$ hold. Since $\dot d(t) < 1$ holds, there exists a unique $\tau_0$ such that $t - d(t) \geq 0$ for all $t \geq \tau_0$. By the continuity of $\mathcal F$, when the social state $x(t)$ converges to $\bar x$, so does the payoff vector $p(t)$ to $\mathcal F(\bar x)$.

  As an example of \eqref{eq:population_game_time_delay}, recall the congestion game from Example~\ref{example:congestion_game}. Suppose an authority collects data on road congestion and its dissemination is subject to a time delay, $d(t)$. The delay may be associated with time it takes to collect and process sufficient traffic data. The underlying payoff mechanism can be formulated as \eqref{eq:population_game_time_delay}. We assume that $d(t)$ is unknown to the agents, but they have estimated bounds $[0, B_d]$ and $[-B_{\dot d}, 1)$ for $d(t)$ and $\dot d(t)$, respectively.

  \begin{remark}
    In \eqref{eq:population_game_time_delay}, the payoff function is subject to a single delay. In practical scenarios, such as the congestion game, the time delay would vary across available strategies. As an immediate extension of \eqref{eq:population_game_time_delay}, we can consider the payoff function with multiple time-dependent delays formulated as
    \begin{align} \label{eq:population_game_multiple_time_delay}
      p(t) = \textstyle \sum_{i=1}^{N_d} \mathcal F_i (x( t - d_i(t))),
    \end{align}
    where $d_i(t) \in [0, B_{d_i}]$ and $\dot d_i(t) \in [-B_{\dot d_i}, 1)$ for fixed positive constants $B_{d_i}$ and $B_{\dot d_i}$, respectively. The payoff function $\mathcal F$ of an underlying population game is recovered when $d_i(t) = 0, ~\forall i \in \{1, \cdots, N_d\}$, i.e., $\mathcal F (x(t)) = \sum_{i=1}^{N_d} \mathcal F_i (x(t))$. Similar analysis established for \eqref{eq:population_game_time_delay} can be applied for \eqref{eq:population_game_multiple_time_delay}.
  \end{remark}

\paragraph{Smoothing Payoff Dynamics Model} \label{section:smoothing_pdm}
We adopt arguments similar to those in \cite[\S V]{Park2019Payoff-Dynamic-} to derive the smoothing PDM \cite{9029756}. Suppose that each agent in population~$k$ has the opportunity for strategy revision at every arrival time of an independent and identically distributed (i.i.d.) Poisson process with parameter~$1$. Let $t$ and $t+h$ be two consecutive strategy revision times for agents in population~$k$. Note that, since the agents' strategy revision is governed by i.i.d. Poisson processes, $h$ approaches zero as $N^k$, the number of agents in population~$k$, tends to infinity. At time $t+h$, the revising agent evaluates the payoff $\mathcal F_j^k (x(t+h))$ of its current strategy~$j$ and shares this information across population~$k$. Subsequently, all agents in the population update \textit{payoff estimates} for available strategies as follows:\footnote{The estimation of the vector $(\mathcal F_1^k (x(t)), \cdots, \mathcal F_{n^k}^k(x(t)))$ is required because the population receives the payoff associated with only one of the strategies that the revising agent currently selects at each revision time $t$.}
\begin{align} \label{eq:payoff_estimation}
  P_i^k(t \!+\! h) \!=\! \begin{cases}
    P_i^k(t) \!+\! h \lambda \frac{\mathcal F_i^k ( x(t+h) ) - P_i^k(t)}{x_i^k(t+h)} & \!\text{if $i = j$} \\
    P_i^k(t) & \!\text{otherwise.}
  \end{cases}
\end{align}
The variable $P_i^k(t)$ is the estimate of $\mathcal F_i^k (x(t))$, and the parameter $\lambda$ is the estimation gain. In expectation, \eqref{eq:payoff_estimation} satisfies
\begin{align} \label{eq:payoff_estimation_expectation}
  \mathbb E \left( \frac{P_i^k(t \!+\! h) \!-\! P_i^k(t) }{h} \right) \!=\! -\lambda \mathbb E ( P_i^k(t) \!-\! \mathcal F_i^k ( x(t \!+\! h) ) ).
\end{align}
For a large number of agents, i.e., as $N^k$ tends to infinity, we can approximate \eqref{eq:payoff_estimation_expectation} with the following equation:
\begin{align} \label{eq:smoothing_pdm}
  \dot p_i^k(t) = -\lambda \left( p_i^k (t) - \mathcal F_i^k \left( x(t) \right) \right).
\end{align}
The variable $p_i^k(t)$ can be viewed as an approximation of $\mathbb E ( P_i^k (t) )$. We refer to \eqref{eq:smoothing_pdm} as the \textit{smoothing PDM}. Note that \eqref{eq:smoothing_pdm} can be interpreted as a low-pass filter applied to the signal $\mathcal F_i^k ( x(t) ), ~ t \geq 0$ and the filtering causes a time delay in computing the payoff estimates. Consequently, the filter output $p_i^k(t)$ lags behind the input $\mathcal F_i^k ( x(t) )$.

  To illustrate the application of the smoothing PDM, consider the two-population zero-sum game described in Example~\ref{example:zero_sum_game}. Here, when an agent in each population~$k$ receives the opportunity to revise its strategy at time~$t+h$, it evaluates the performance of its current strategy~$j$ in terms of the payoff $\mathcal F_j^k ( x(t+h) )$  and shares this with other agents in the same population. All agents in population~$k$ then update their payoff estimates according to \eqref{eq:payoff_estimation}, and a solution to the smoothing PDM \eqref{eq:smoothing_pdm} approximates these payoff estimates.

  \change{
    To formalize the class of population games studied in this paper, we adopt the causal mapping notation $\mathfrak G$, which maps a population state trajectory $x(t), \, t \geq 0$ to a payoff vector trajectory $p(t), \, t \geq 0$. The formal definition of $\mathfrak G$ is provided in Appendix~\ref{def:causal_mapping}. This mapping encapsulates the dependency of the payoff vector $p(t)$ at time $t$ on the social states from current and prior times, represented as $x(\tau)$ for $0 \leq \tau \leq t$. This dependency is expressed by the following relationship:
  \begin{align} \label{eq:pdm}
    p(t) = \left( \mathfrak G (x) \right) (t),
  \end{align}
where we represent the population state trajectory as a time-dependent function $x: \mathbb R_{\geq 0} \to \mathbb X$.}
  We relate the mapping $\mathfrak G$ to the payoff function $\mathcal F$ of an underlying population game by requiring that when $x(t)$ converges to a steady state, i.e., $\lim_{t \to \infty} x(t) = \bar x$, $p(t)$ converges to $\mathcal F(\bar x)$, i.e., $\lim_{t \to \infty} p(t) = \mathcal F(\bar x)$.
  Following the same naming convention as in \cite{9029756}, we refer to \eqref{eq:pdm} as the \textit{payoff dynamics model (PDM)}. 

  \begin{remark}
    The original definition \cite[Definition~7]{9029756} of the PDM, described as a finite-dimensional dynamical system, is 
    \begin{subequations} \label{eq:pdm_finite_dimensional}
      \begin{align}
        \dot q(t) &= \mathcal G(q(t), x(t)) \\
        p(t) &= \mathcal H(q(t), x(t)),
      \end{align}
    \end{subequations}
    where $q(t) \in \mathbb R^n$ represents an internal state of the PDM. The definition \eqref{eq:pdm} of the PDM adopted in this paper expands on \eqref{eq:pdm_finite_dimensional} to represent certain types of payoff mechanisms, such as those in \eqref{eq:population_game_time_delay}, that cannot be represented by \eqref{eq:pdm_finite_dimensional}.
  \end{remark}

We make the following assumption on \eqref{eq:pdm}.
\begin{assumption} \label{assumption:pdm}
  The PDM \eqref{eq:pdm} satisfies the following.
  \begin{enumerate} 
    \item \label{assumption:stationary_model} Suppose $\mathcal F$ is the payoff function of an underlying population game and the given social state $x(t)$ is differentiable. The PDM \eqref{eq:pdm} satisfies $\lim_{t \to \infty} \| \dot x(t) \|_2 = 0 \implies \lim_{t \to \infty} \| p(t) - \mathcal F(x(t)) \|_2 = 0,$ where $p(t)$ is the payoff vector determined by \eqref{eq:pdm} with  $x(t), ~ t \geq 0$.

  \item Given a social state trajectory $x(t), ~ t \geq 0$, \eqref{eq:pdm} computes a unique payoff vector trajectory $p(t), ~ t \geq 0$. In other words, for any pair of social state trajectories $x(t), ~ t \geq 0$ and $y(t), ~ t \geq 0$, it holds that $x(t) = y(t), ~ \forall t \geq 0 \implies ( \mathfrak G (x) )(t) = ( \mathfrak G (y) )(t), ~ \forall t \geq 0.$

  \item If the social state $x(t)$ is differentiable, so is the resulting payoff vector $p(t)$. Given that $\dot x(t)$ is bounded, both $p(t)$ and $\dot p(t)$ are bounded, i.e., there exist $B_p$ and $B_{\dot p}$ satisfying $\| p(t) \|_2 \leq B_p$ and $\| \dot p(t) \|_2 \leq B_{\dot p}$ for all $t \geq 0$, respectively.
  \end{enumerate}
\end{assumption}
Note that both the payoff function with a time-dependent delay \eqref{eq:population_game_time_delay} and the smoothing PDM \eqref{eq:smoothing_pdm} satisfy Assumption~\ref{assumption:pdm}.\footnote{In particular, by using the mean value theorem and Assumption~\ref{assumption:payoff_function}, we can verify that \eqref{eq:population_game_time_delay} satisfies Assumptions~\ref{assumption:pdm}-1 and \ref{assumption:pdm}-3. Additionally, by applying the same arguments as those used in the proof of \cite[Proposition~6]{Park2019Payoff-Dynamic-}, we can confirm that \eqref{eq:smoothing_pdm} meets Assumptions~\ref{assumption:pdm}-1 and \ref{assumption:pdm}-3.} In light of Assumption~\ref{assumption:pdm}-\ref{assumption:stationary_model}, as originally suggested in \cite{Fox2013Population-Game}, we can view the PDM \eqref{eq:pdm} as a \textit{dynamic modification} of the conventional population game model.

  \subsection{Strategy Revision and Evolutionary Dynamics Model} \label{sec:decision_making_model}
  Suppose each agent in population~$k$ revises its strategy selection at each arrival time $t$ of an i.i.d. Poisson process with parameter~$1$, where the strategy revision depends on the payoff vector $p^k(t)$ and population state $x^k(t)$. We adopt the evolutionary dynamics framework \cite[Part~II]{Sandholm2010-SANPGA-2} in which the change of the population state $x^k(t)$ when the number of agents $N^k$ in the population tends to infinity is described by the following ordinary differential equation:
  For $i$ in $\{ 1, \cdots, n^k \}$ and $k$ in $\{1, \cdots, M \}$,
  \begin{multline} \label{eq:edm}
    \dot x_i^k(t) = \textstyle\sum_{j=1}^{n^k} x_j^k(t) \mathcal T_{ji}^k ( x^k(t), p^k(t) ) \\
    - x_i^k(t) \textstyle\sum_{j=1}^{n^k} \mathcal T_{ij}^k ( x^k(t), p^k(t) ),
  \end{multline}
  where the payoff vector $p^k(t)$ is determined by the PDM \eqref{eq:pdm}. The \textit{strategy revision protocol} $\mathcal T_{ji}^k (z^k, r^k)$ defines the probability that each agent in population~$k$ switches its strategy from $j$ to $i$ when the population state and payoff vector take the values $z^k \in \mathbb X^k$ and $r^k \in \mathbb R^{n^k}$, respectively.\footnote{The reference \cite[\S 5]{Sandholm2010-SANPGA-2} summarizes well-known protocols developed in the game theory literature. Also, we refer the interested reader to \cite[\S 10]{Sandholm2010-SANPGA-2} and \cite[\S IV]{9029756} for the derivation of \eqref{eq:edm}.} As in \cite{9029756}, we refer to \eqref{eq:edm} as the \textit{evolutionary dynamics model (EDM)}.
  Among existing strategy revision protocols, the most relevant to our study is the \textit{logit protocol} defined as
  \begin{align} \label{eq:StandardLogitProtocol}
    \mathcal T_{i}^{\text{\scriptsize Logit}} ( r^k ) = \frac{\exp ( \eta^{-1} r_i^k )}{\sum_{l=1}^{n^k} \exp ( \eta^{-1} r_l^k )},
  \end{align}
  where $\eta$ is a positive constant and $r^k = (r_1^k, \cdots, r_{n^k}^k) \in \mathbb R^{n^k}$ is the value of population~$k$'s payoff vector. The agents adopting the logit protocol, i.e., $\mathcal T_{ji}^k ( z^k, r^k ) = \mathcal T_{i}^{\text{\scriptsize Logit}} ( r^k )$, revise their strategy choices based only on the payoff vector and the probability of switching to strategy $i$ is independent of current strategy $j$.

  \begin{figure} [t]
    \center
    \subfigure[$\eta=0.1$]{
      \includegraphics[trim={.0in .2in 0 .1in},width=1.6in]{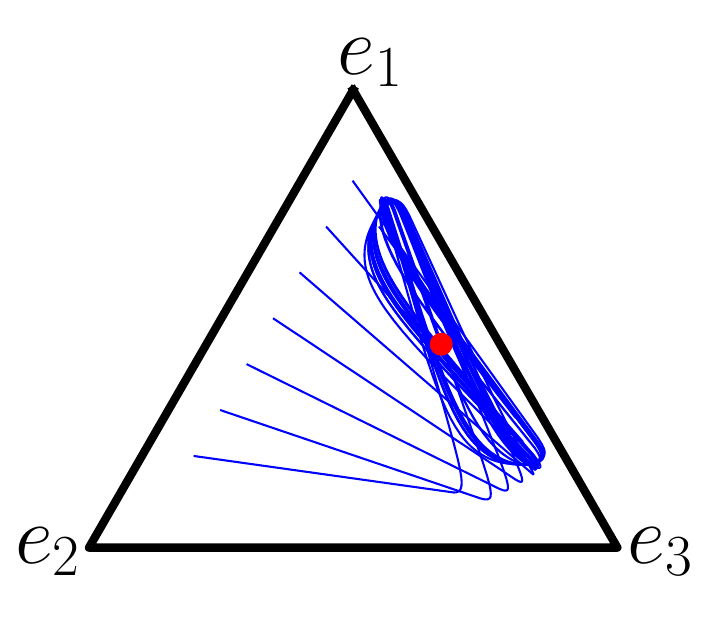}
      \label{fig:standard_logit_models_delayed_payoff_a}
    }
    \subfigure[$\eta=4.5$]{
      \includegraphics[trim={.0in .2in 0 .1in},width=1.6in]{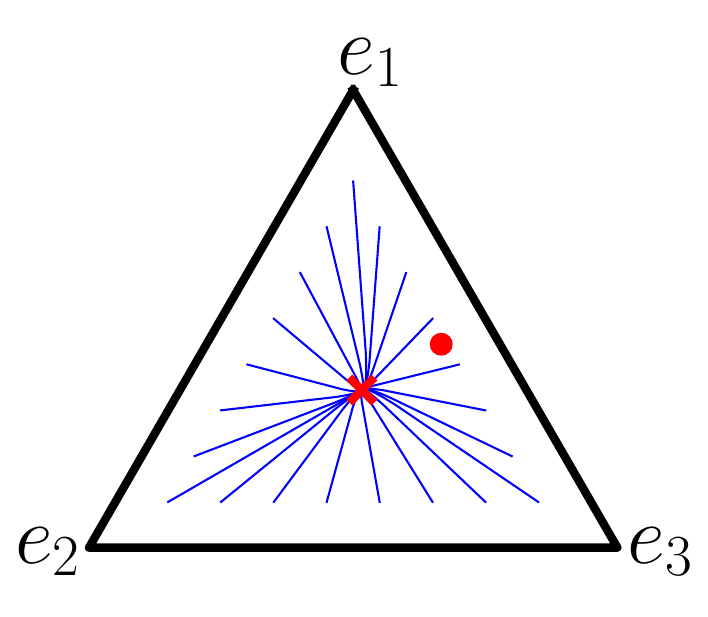}
      \label{fig:standard_logit_models_delayed_payoff_b}
    }

    \caption{State trajectories of population~$1$ under the logit protocol \eqref{eq:StandardLogitProtocol} with $\eta=0.1, 4.5$ in the congestion game \eqref{eq:congestion_game}. The payoff vector is determined by \eqref{eq:population_game_time_delay} subject to a fixed unit time delay  ($d(t)=1, ~ \forall t \geq 0$). The red circle in both (a) and (b) marks the Nash equilibrium and the red X mark in (b) denotes the unique limit point of all the trajectories.
    }
    \label{fig:standard_logit_models_delayed_payoff}
    \vspace{-.15in}
  \end{figure}
  \begin{figure} [t]
    \center
    \subfigure[$\eta=0.1$]{
      \includegraphics[trim={.0in .2in 0 .1in},width=1.6in]{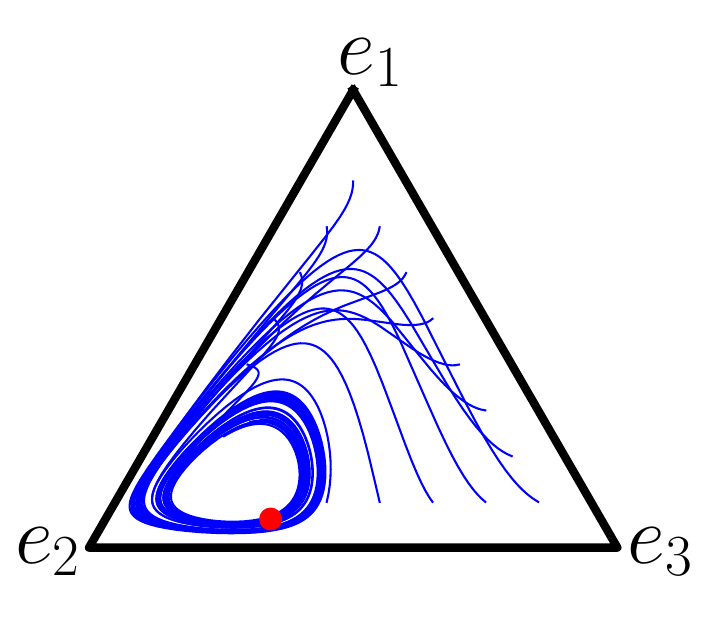}
      \label{fig:standard_logit_models_smoothing_pdm_a}
    }
    \subfigure[$\eta=0.6$]{
      \includegraphics[trim={.0in .2in 0 .1in},width=1.6in]{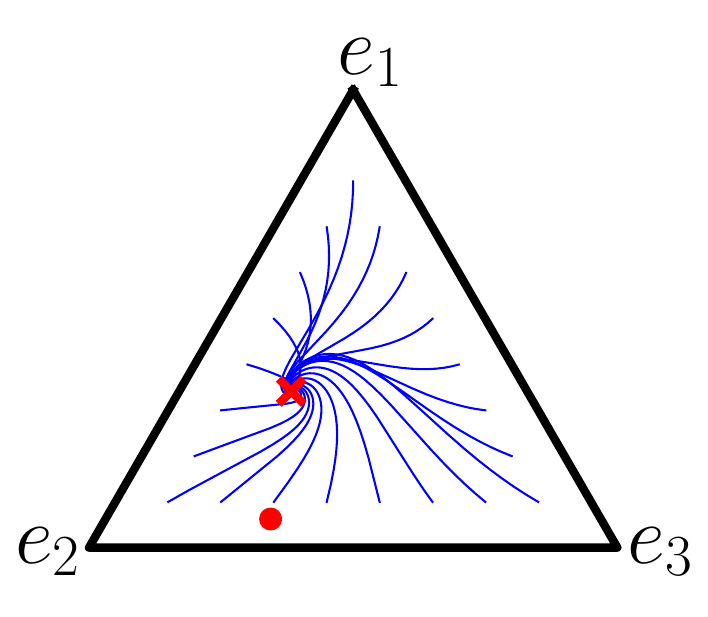}
      \label{fig:standard_logit_models_smoothing_pdm_b}
    }

    \caption{State trajectories of population~$1$ under the logit protocol \eqref{eq:StandardLogitProtocol} with $\eta=0.1, 0.6$ in the zero-sum game \eqref{eq:rps_game}. The payoff vector is determined by  \eqref{eq:smoothing_pdm} with $\lambda=1$. The red circle in both (a) and (b) marks the Nash equilibrium and the red X mark in (b) denotes the unique limit point of all the trajectories.
    }
    \label{fig:standard_logit_models_smoothing_pdm}
    \vspace{-.2in}
  \end{figure}

  When the payoff mechanisms are subject to time delays as in \eqref{eq:population_game_time_delay} or \eqref{eq:smoothing_pdm}, under existing EDMs, the social state tends to oscillate or converge to equilibrium points that are different from the Nash equilibrium. To illustrate this,
  using the logit protocol \eqref{eq:StandardLogitProtocol} with two different values of $\eta$, Figs.~\ref{fig:standard_logit_models_delayed_payoff} and \ref{fig:standard_logit_models_smoothing_pdm} depict social state trajectories derived from \eqref{eq:edm} in the congestion game \eqref{eq:congestion_game} where the payoff function is subject to a unit time delay \eqref{eq:population_game_time_delay} and in the zero-sum game \eqref{eq:rps_game} where the payoff vector is defined by the smoothing PDM \eqref{eq:smoothing_pdm}, respectively. We observe that when $\eta$ is small, the resulting social state trajectories oscillate, whereas with a sufficiently large $\eta$, the trajectories converge to an equilibrium point which is located away from the Nash equilibrium.

  To improve the limitations of existing protocols, we propose a new strategy revision protocol and analyze its convergence properties to rigorously show that the new protocol allows all agents to asymptotically attain the Nash equilibrium even with time delays in payoff mechanisms.
  We formally state the main problem as follows.

\noindent \rule{\columnwidth}{1pt}
\begin{problem} \label{problem}
  Design a strategy revision protocol $\mathcal T_{ji}^k$ and find conditions on the PDM \eqref{eq:pdm} and EDM \eqref{eq:edm} under which the social state $x(t)$ converges to the Nash equilibrium set $\mathbb{NE} \left( \mathcal F \right)$:
  \begin{align*}
    \lim_{t \to \infty} \inf_{z \in \mathbb{NE} (\mathcal F)} \|x(t) - z\|_2 = 0.
  \end{align*}
\end{problem}
\noindent \rule{\columnwidth}{1pt}

\section{Literature Review} \label{sec:literature_review}

Multi-agent decision problems formalized by the population games and evolutionary dynamics framework have been a major research theme in control systems and neighboring research communities due to their importance in a wide range of applications \cite{7422154, 9561809, doi:10.1080/00207179.2010.501389, 6502215, 5345810, 9304437, 6459524}. As well-documented in \cite{hofbauer_sigmund_1998, Sandholm2010-SANPGA-2}, the framework has a long and concrete history of research endeavors in developing Lyapunov stability-based techniques to analyze the long-term behavior of decision-making models.

The authors of \cite{Fox2013Population-Game} present pioneering work in generalizing the conventional population games formalism to include \textit{dynamic} payoff mechanisms, an earlier form of the PDM, and in exploring the use of passivity from dynamical system theory \cite{willems_dissipative_1972} for stability analysis, unifying notable stability results in the game theory literature, e.g., \cite{RePEc:eee:jetheo:v:144:y:2009:i:4:p:1665-1693.e4}. In subsequent works, such as \cite{9483414, MARTINEZPIAZUELO2022110227}, the PDM formalism has been adopted to model time delays in payoff mechanisms, as in our work, and also to design payoff mechanisms that incentivize decision-making agents to learn and attain the \textit{generalized Nash equilibrium}, i.e., the Nash equilibrium satisfying given constraints on agent decision-making.

Further studies have led to a more concrete characterization of stability and the development of passivity-based tools for convergence analysis in population games. The tutorial article \cite{9029756} and its supplementary material \cite{Park2019Payoff-Dynamic-} detail such formalization and technical discussions on $\delta$-passivity in population games and a wide class of EDMs. The authors of \cite{9219202} discuss a more general framework -- dissipativity tools -- for convergence analysis and explain its importance in analyzing road congestion with mixed autonomy.\footnote{Although adopting the dissipativity tool of \cite{9219202} would lead to more general discussions on convergence analysis, for conciseness, we adopt the passivity-based approaches \cite{Fox2013Population-Game, 9029756} as these are sufficient to establish our main results.} We also refer the interested reader to \cite{9781277, 9022871} for different applications of passivity/dissipativity theory in finite games.

There is a substantial body of literature that investigates the effect of time delays in multi-agent games \cite{ALBOSZTA2004175, oaku2002, bodnar2020, Khalifa2018, IIJIMA20121, YI1997111, tembine2011, wesson2016, PhysRevE.101.042410, WANG20178, 7039982, Obando2016, doi:10.1142/S0218127413501228, 10.5555/1345263.1345309, HU2019218, RePEc:eee:matsoc:v:61:y:2011:i:2:p:83-85}. These references, as we also illustrated in Figs.~\ref{fig:standard_logit_models_delayed_payoff_a} and \ref{fig:standard_logit_models_smoothing_pdm_a}, explain that such time delays result in the oscillation of state trajectories. In particular, the references \cite{YI1997111, tembine2011, ALBOSZTA2004175, wesson2016, 10.5555/1345263.1345309, Obando2016, 7039982} discuss the stability of the replicator dynamics in population games defined by affine payoff functions that are subject to time delays.
Notably, \cite{wesson2016, PhysRevE.101.042410, doi:10.1142/S0218127413501228} adopt Hopf bifurcation analysis to rigorously show that oscillation of state trajectories emerges as the time delay increases. Stability and bifurcation analysis on other types of EDMs, such as the best response dynamics and imitation dynamics, in population games with time delays are investigated in \cite{tembine2011, oaku2002, PhysRevE.101.042410, HU2019218, RePEc:eee:matsoc:v:61:y:2011:i:2:p:83-85, WANG20178}. Whereas these works regard the time delays as deterministic parameters in their stability analysis, others \cite{Khalifa2018, IIJIMA20121, HU2019218} study the stability of the Nash equilibrium when the time delays are defined by random variables with exponential, uniform, or discrete distributions.

  Regularization in designing agent decision-making models has been explored in multi-agent games. References \cite{10.1137/S0363012903437976, Mertikopoulos2016, MERTIKOPOULOS2018315, 10.2307/24540967, 9022871, 9781277} adopt regularization to design reinforcement learning models and discuss convergence to the Nash equilibrium in finite games. \cite{10.1137/S0363012903437976} presents an earlier work of so-called exponentially discounted reinforcement learning -- later further developed in \cite{10.2307/24540967, 9022871} -- and discusses how the regularization improves the convergence of the learning dynamics.
  The authors of \cite{Mertikopoulos2016} provide extensive discussions on reinforcement learning models in finite games including the exponentially discounted reinforcement learning model.
  \cite{MERTIKOPOULOS2018315} explores the use of regularization in population game settings and proposes \textit{Riemannian game dynamics}, where the regularization is defined by a Riemannian metric such as the Euclidean norm.
  More recent works \cite{9022871, 9781277, 9029756} explain the benefit of the regularization using passivity theory,
  rigorously showing that the regularization in agent decision-making models enhances the models' passivity measure.

Unlike existing studies on the effect of time delays in population games, which focus on identifying technical conditions under which oscillation of state trajectories emerges, we propose a new model that guarantees the convergence of the social state to the Nash equilibrium in a class of contractive population games that are subject to time delays. Although the idea of regularization and passivity-based analysis have been reported in the multi-agent games literature, all previous results only establish the convergence to the ``perturbed'' Nash equilibrium.

  This paper substantially extends our earlier work \cite{9483414} by generalizing convergence results for the KLD-RL model to a broader class of PDMs, as discussed in \S\ref{sec:regularizedLogit}. For example, we apply the main results to scenarios where the payoff vector is determined by the payoff function with a time-dependent delay and the smoothing PDM.
  In contrast, the conference version \cite{9483414} only discusses results for a payoff function with a fixed delay. Additionally, this paper extends \cite{9483414} to multi-population scenarios and proposes a distributed approach that allows individual agents to update regularization parameters of the new model, as explained in \S\ref{section:distributed_parameter_update}. In this approach, agents only need to assess the state of the population to which they belong and the payoff vector associated with their available strategies.  The approach also accounts for the fact that the information exchange between agents is limited according to a given communication graph.

\section{Learning Nash Equilibrium with \\ Delayed Payoffs}
\label{sec:regularizedLogit}
Given $z^k = (z_1^k, \cdots, z_{n^k}^k)$ and $\theta^k = (\theta_1^k, \cdots, \theta_{n^k}^k)$, both belonging to the interior of the population state space $\mathrm{int}(\mathbb{X}^k)$, we define the Kullback-Leibler divergence (KLD) as
\begin{align} \label{eq:kl_divergence}
  \mathcal{D}(z^k \,\|\, \theta^k) = \textstyle\sum_{i=1}^{n^k} z_i^k \ln \frac{z_i^k}{\theta_i^k}.
\end{align}
We compute the gradient of \eqref{eq:kl_divergence} in $\mathbb{X}^k$ (with respect to the first argument $z^k$) as
\begin{align} \label{eq:gradient_kld}
  \nabla \mathcal{D}(z^k \,\|\, \theta^k) = \begin{pmatrix} \ln \frac{z_1^k}{\theta_1^k} & \cdots & \ln \frac{z_{n^k}^k}{\theta_{n^k}^k} \end{pmatrix}'.
\end{align}
Note that \eqref{eq:kl_divergence} is a convex function of $z^k$. For notational convenience, we use
$\mathcal{D}(z \,\|\, \theta) = \sum_{k=1}^{M} \mathcal{D}(z^k \,\|\, \theta^k)
$ and
$\nabla \mathcal{D}(z \,\|\, \theta) = (\nabla \mathcal{D}(z^1 \,\|\, \theta^1), \cdots, \nabla \mathcal{D}(z^M \,\|\, \theta^M)).$

For given $\theta^k \in \mathrm{int}(\mathbb{X}^k)$, using \eqref{eq:kl_divergence}, we define the KLD Regularized Learning (KLD-RL) protocol $\mathcal{T}^{\text{\scriptsize KLD-RL}}(\theta^k, r^k) = (\mathcal{T}_1^{\text{\scriptsize KLD-RL}}(\theta^k, r^k), \cdots, \mathcal{T}_{n^k}^{\text{\scriptsize KLD-RL}}(\theta^k, r^k))$ that maximizes a regularized average payoff:
\begin{align} \label{eq:maximization}
  \mathcal{T}^{\text{\scriptsize KLD-RL}}(\theta^k, r^k) = \argmax_{z^k \in \mathrm{int}(\mathbb{X}^k)} ((z^k)' r^k - \eta \mathcal{D}(z^k \,\|\, \theta^k)),
\end{align}
where $r^k \in \mathbb{R}^{n^k}$ is the value of population $k$'s payoff vector, and $\eta > 0$ is a weight on the regularization. Under this protocol, agents revise their strategies to maximize the cost in \eqref{eq:maximization}, combining the average payoff $(z^k)' r^k$ with the regularization $\mathcal{D}(z^k \,\|\, \theta^k)$ weighted by $\eta$.

Following a similar argument as in \cite{10.2307/3081987}, we can identify a unique solution to \eqref{eq:maximization} as
\begin{align} \label{eq:kld_rl}
  \mathcal{T}_i^{\text{\scriptsize KLD-RL}}(\theta^k, r^k) = \frac{\theta_i^k \exp(\eta^{-1} r_i^k)}{\sum_{l=1}^{n^k} \theta_l^k \exp(\eta^{-1} r_l^k)}.
\end{align}
A key aspect of the KLD-RL protocol \eqref{eq:kld_rl} is using $\theta^k$ as a tuning parameter. As a special case, when $\theta^k = x^k$, $\eqref{eq:kld_rl}$ becomes the imitative logit protocol \cite[Example~5.4.7]{Sandholm2010-SANPGA-2}. In \S\ref{sec:iterative_regularization}, we propose an algorithm to compute an appropriate value for $\theta^k$ that guarantees the convergence of the social state to the Nash equilibrium set.

To further discuss the effect of $\theta^k$ on agents' strategy revision, consider the following two special cases:
\begin{itemize}
\item \textbf{Case I:} If $r_1^k = \cdots = r_{n^k}^k$, then
  \begin{align*}
    \mathcal{T}_i^{\text{\scriptsize KLD-RL}}(\theta^k, r^k) = \theta_i^k.
  \end{align*}
\item \textbf{Case II:} If $\theta_1^k = \cdots = \theta_{n^k}^k$, then
  \begin{align*}
    \mathcal{T}_i^{\text{\scriptsize KLD-RL}}(\theta^k, r^k) = \frac{\exp(\eta^{-1} r_i^k)}{\sum_{l=1}^{n^k} \exp(\eta^{-1} r_l^k)}.
  \end{align*}
\end{itemize}
From (\textbf{Case I}), we observe that $\theta^k$ serves as a bias in the agents' strategy revision. When the elements of the payoff vector are identical across all strategies, agents tend to select strategies with a higher value of $\theta_i^k$. When there is no bias (\textbf{Case II}), \eqref{eq:kld_rl} is equivalent to the logit protocol \eqref{eq:StandardLogitProtocol}.

\begin{figure} [t]
  \centering
  \includegraphics[trim={.0in .0in 0 0}, width=3.2in]{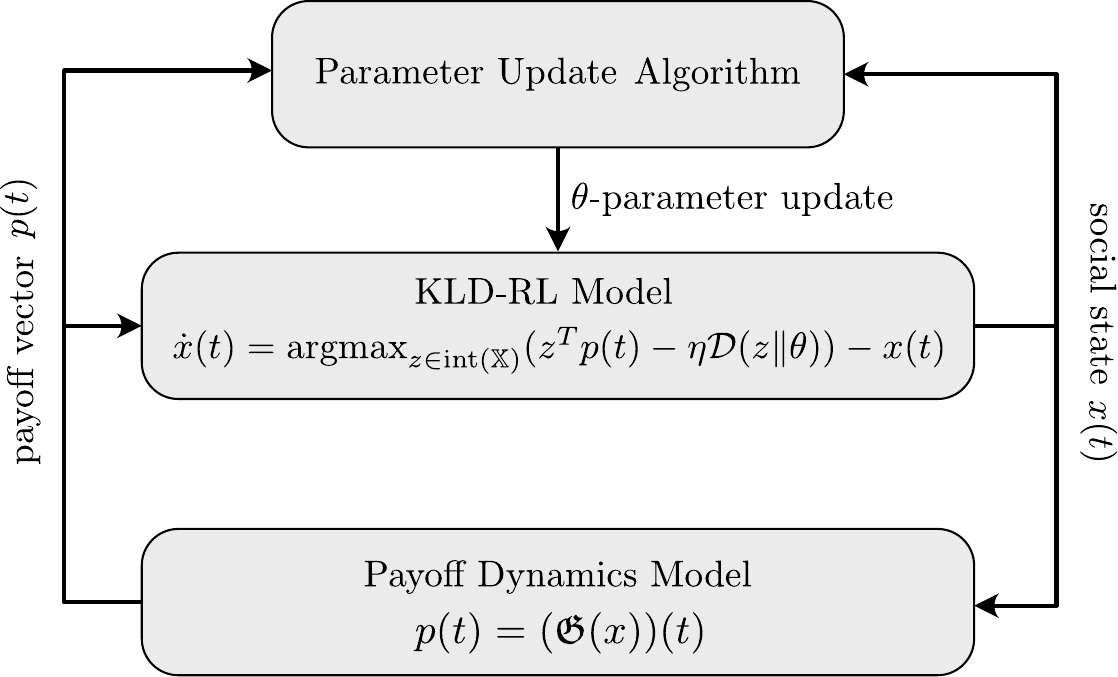}
  \caption{A feedback interconnection illustrating the payoff dynamics model and the KLD-RL model along with an algorithm for updating the regularization parameter $\theta$.} 
  \label{fig:closed_loop}
\end{figure}

  To study the asymptotic behavior of the EDM, as defined by the KLD-RL protocol in \eqref{eq:kld_rl}, i.e., $\mathcal{T}_{ji}^k (z^k, r^k) = \mathcal{T}_{i}^{\text{\scriptsize KLD-RL}} (r^k)$, we express the state equation of the closed-loop model consisting of \eqref{eq:pdm} and \eqref{eq:edm} as follows: For $i \in \{1, \cdots, n^k\}$ and $k \in \{1, \cdots, M\}$, 
  \begin{subequations} \label{eq:closed_loop}
    \begin{align} 
      \dot x_i^k(t) &= \frac{\theta_i^k \exp(\eta^{-1} p_i^k(t))}{\sum_{l=1}^{n^k} \theta_l^k \exp(\eta^{-1} p_l^k(t))} - x_i^k(t) \label{eq:closed_loop_a} \\
      p(t) &= (\mathfrak{G} (x))(t). \label{eq:closed_loop_b}
    \end{align}
  \end{subequations}
We assume that given an initial condition $(x(0), p(0)) \in \mathbb{X} \times \mathbb{R}^n$, there is a unique solution to \eqref{eq:closed_loop}. Since $p(t)$ is bounded by Assumption~\ref{assumption:pdm}-3, the social state $x(t)$ belongs to $\mathrm{int}(\mathbb{X})$ for all $t \geq 0$. 

In \S\ref{sec:preliminary_convergence_results}, we present preliminary convergence analysis for \eqref{eq:closed_loop} with parameter $\theta$ fixed. Then, in \S\ref{sec:iterative_regularization}, we propose a parameter update algorithm that ensures the convergence of the social state $x(t)$ derived by \eqref{eq:closed_loop} to the Nash equilibrium set. Fig.~\ref{fig:closed_loop} illustrates our framework consisting of \eqref{eq:closed_loop} and a parameter update algorithm for $\theta = (\theta^1, \cdots, \theta^M)$.

\subsection{Preliminary Convergence Analysis} \label{sec:preliminary_convergence_results}

Our analysis hinges on the passivity technique developed in \cite{Fox2013Population-Game, Park2019Payoff-Dynamic-, 9029756}. We begin by reviewing two notions of passivity -- \textit{weak $\delta$-antipassivity} and \textit{$\delta$-passivity} -- adopted for \eqref{eq:pdm} and \eqref{eq:edm}, respectively. We then establish the stability of the closed-loop model \eqref{eq:closed_loop}.

\begin{definition} [Weak $\delta$-Antipassivity with Deficit $\bar \nu$ \cite{9029756}] \label{def:weak_delta_passivity}
  The PDM \eqref{eq:pdm} is \textit{weakly $\delta$-antipassive with deficit $\bar \nu$} if \textbf{(i)} there exists a positive and bounded function $\alpha_{x,p}: \mathbb R_{\geq 0} \to \mathbb R_{\geq 0}$ such that
  \begin{align} \label{eq:weakAntipassivity}
    \alpha_{x,p} (t_0) \!\geq\! \int_{t_0}^t \!( \dot x'\!(\tau) \dot p(\tau) \!-\! \nu \dot x'\!(\tau) \dot x(\tau) ) \, \mathrm d \tau, \,\forall t \!\geq\! t_0 \!\geq\! 0
  \end{align}
  holds for every social state trajectory $x(t), \, t \geq 0$ and for every nonnegative constant $\nu \!>\! \bar \nu$, and \textbf{(ii)} $\lim_{t \to \infty} \alpha_{x,p}(t) \!=\! 0$ holds whenever $\lim_{t \to \infty} \|\dot x(t)\|_2 \!=\! 0.$\footnote{We use the subscript to indicate that the function $\alpha_{x,p}$ depends on the trajectories $x(t), \, t \geq 0$ and $p(t), \, t \geq 0$. We note that the requirement $\lim_{t \to \infty} \|\dot x(t)\|_2 = 0 \implies \lim_{t \to \infty} \alpha_{x,p}(t) = 0$ is an addition to the original definition given in \cite{9029756}. Hence, when the PDM satisfies Definition~\ref{def:weak_delta_passivity}, it also satisfies the original definition from \cite{9029756}.} The payoff vector trajectory $p(t), \, t \geq 0$ in \eqref{eq:weakAntipassivity} is derived by \eqref{eq:pdm} and given $x(t), \, t \geq 0$.
\end{definition}

The constant $\bar \nu$ is a measure of the passivity deficit in \eqref{eq:pdm}. Viewing \eqref{eq:pdm} as a dynamical system, according to \cite{willems_dissipative_1972}, the function $\alpha_{x,p}$ is an estimate of the \textit{stored energy} of \eqref{eq:pdm}.

In the following lemmas, we establish the weak $\delta$-antipassivity of the payoff function with a time-dependent delay \eqref{eq:population_game_time_delay} and the smoothing PDM \eqref{eq:smoothing_pdm}. The proofs of the lemmas are provided in Appendix~\ref{proof_lemma:antipassivity}.
\begin{lemma} \label{lemma:antipassivity_delayed_pdm}
  The payoff function with a time-dependent delay \eqref{eq:population_game_time_delay} is weakly $\delta$-antipassive with a positive deficit $\bar \nu = \frac{B_{D \mathcal F}(2+B_{\dot d})}{2}$, where $B_{D \mathcal F}$ is the upper bound on $D\mathcal F$ as defined in Assumption~\ref{assumption:payoff_function} and $B_{\dot d}$ is defined such that $\dot d(t) \in [-B_{\dot d}, 1)$ holds. Additionally, the function $\alpha_{x,p}$ for \eqref{eq:weakAntipassivity} is defined as
  \begin{align} \label{eq:alpha_delayed_pdm}
    \alpha_{x,p} (t_0) = \frac{B_{D\mathcal F}(1+B_{\dot d})}{2} \int_{t_0-B_d}^{t_0} \| \dot{x} (\tau) \|_2^2 \, \mathrm d\tau,
  \end{align}
  where $B_d$ is the upper bound of $d(t)$, i.e., $d(t) \in [0, B_d]$.
\end{lemma}

Note that by \eqref{eq:alpha_delayed_pdm}, it holds that $\lim_{t \to \infty} \|\dot x(t)\|_2 = 0$ implies $\lim_{t \to \infty} \alpha_{x,p}(t) = 0$.

\begin{lemma} \label{lemma:antipassivity_smoothing_pdm}
  Consider the smoothing PDM \eqref{eq:smoothing_pdm} where its underlying payoff function is contractive and defined by an affine mapping $\mathcal F(x) = F x + b$. The smoothing PDM is weakly $\delta$-antipassive with a positive deficit $\bar \nu$ given by $\bar \nu = \frac{1}{4} \| F - F' \|_2$,
  and $\alpha_{x,p} (t_0)$ in \eqref{eq:weakAntipassivity} is defined as $\alpha_{x,p} (t_0) = \sqrt{n} \| p(t_0) - F x(t_0) - b \|_2,$
  where $n = \sum_{k=1}^M n^k$. 
\end{lemma}

Since \eqref{eq:smoothing_pdm} satisfies Assumption~\ref{assumption:pdm}-1, $\lim_{t \to \infty} \|\dot x(t)\|_2 =0$ implies $\lim_{t \to \infty} \alpha_{x,p}(t) = 0$.
If $F$ is symmetric, then $\bar \nu = 0$ and the smoothing PDM \eqref{eq:smoothing_pdm} becomes weakly $\delta$-antipassive without any deficit ($\bar \nu = 0$), which coincides with \cite[Proposition~7-i]{Park2019Payoff-Dynamic-}. Hence,  Lemma~\ref{lemma:antipassivity_smoothing_pdm} extends \cite[Proposition~7-i]{Park2019Payoff-Dynamic-} to the case where $F$ is non-symmetric. 

\begin{definition} [$\delta$-Passivity with Surplus $\bar \eta$ \cite{9029756}] \label{def:delta_passivity}
  The EDM \eqref{eq:edm} is \textit{$\delta$-passive with surplus $\bar \eta$} if there is a continuously differentiable function $\mathcal S: \mathbb X \times \mathbb R^n \to \mathbb R_{\geq 0}$ for which 
  \begin{multline} \label{eq:delta_passivity}
    \mathcal S \left( x(t), p(t) \right) - \mathcal S \left( x(t_0), p(t_0) \right) \leq  \\
    \int_{t_0}^t \left( \dot x'(\tau) \dot p(\tau) - \eta \, \dot x'(\tau) \dot x(\tau) \right) \, \mathrm d \tau, ~ \forall t \geq t_0 \geq 0
  \end{multline}
  holds for every payoff vector trajectory $p(t), ~ t \geq 0$ and for every nonnegative constant $\eta < \bar \eta$, where the social state trajectory $x(t), ~ t \geq 0$ is determined by \eqref{eq:edm} and given $p(t), ~ t \geq 0$. We refer to $\mathcal S$ as the \textit{$\delta$-storage function}. 
  We call $\mathcal S$ \textit{informative} if the function satisfies
  \begin{align*}
    \mathcal S \left( z, r  \right) \!=\! 0 \iff
    \mathcal V \left( z, r \right) \!=\! 0 \iff
    \nabla_z' \mathcal S \left( z, r \right) \!\mathcal V \left( z, r \right) \!=\! 0,
  \end{align*}
  where $\mathcal V = \left( \mathcal V^1, \cdots, \mathcal V^M \right)$ denotes the vector field of \eqref{eq:edm} defining $\dot x^k(t) = \mathcal V^k(x^k(t), p^k(t))$.
\end{definition}

The constant $\bar \eta$ is a measure of passivity surplus in \eqref{eq:edm}. Compared to weak $\delta$-antipassivity defined for \eqref{eq:pdm}, Definition~\ref{def:delta_passivity} states a stronger notion of passivity since it requires the existence of the $\delta$-storage function $\mathcal S$.

In the following lemma, we establish $\delta$-passivity of the KLD-RL EDM \eqref{eq:closed_loop_a}.
\begin{lemma} \label{lemma:regLogitPassivity}
  Given fixed weight $\eta > 0$ and regularization parameter $\theta \in \mathrm{int} (\mathbb X)$, the KLD-RL EDM \eqref{eq:closed_loop_a} is $\delta$-passive with surplus $\eta$ and has an informative $\delta$-storage function $\mathcal S_{\theta}: \mathbb X \times \mathbb R^n \to \mathbb R_{\geq 0}$ expressed as\footnote{We use the subscript to specify the dependency of $\mathcal S_{\theta}$ on $\theta$. Also, as we discussed in the beginning of  \S\ref{sec:regularizedLogit}, there is a unique solution, given as in  \eqref{eq:kld_rl}, for the maximization in \eqref{eq:DeltaStorageFunc}.}
  \begin{multline} \label{eq:DeltaStorageFunc}
    \mathcal S_{\theta} ( z, r ) \!=\!
    \max_{{\bar z} \in \mathrm{int}(\mathbb X)} ( \bar z' r \!-\! \eta \mathcal D ( \bar z \| \theta ) ) 
    \!-\! ( z' r \!-\! \eta \mathcal D ( z \| \theta ) ), 
  \end{multline}
  where $z = (z^1\!,\!\cdots\!,\!z^M)$, $r = (r^1\!,\!\cdots\!,\!r^M)$, and $\theta = (\theta^1\!,\!\cdots\!,\!\theta^M)$.
\end{lemma}
  \begin{proof}
    Noting that \eqref{eq:closed_loop_a} belongs to the class of the perturbed best response EDMs \cite[\S VIII]{Park2019Payoff-Dynamic-}, the proof follows from \cite[Proposition~5]{Park2019Payoff-Dynamic-}. \QED
  \end{proof}

  Let $(\bar x, \bar p)$ be a stationary point of \eqref{eq:closed_loop}. By Assumption~\ref{assumption:pdm}, it holds that $\bar p = \mathcal{F}(\bar x)$. Additionally, according to \eqref{eq:maximization} and \cite[Lemma~A.1]{HOFBAUER200747}, we have that $\mathcal F(\bar x) - \eta \nabla \mathcal D (\bar x \,\|\, \theta) = c \mathbf 1$, where $c$ is a constant and $\mathbf 1 \in \mathbb R^n$ is a vector with all elements equal to 1. Consequently, the following relation can be established:
  \begin{align} 
    \max_{z \in \mathbb{X}} (z - \bar x)' (\mathcal{F}(\bar x) - \eta \nabla \mathcal{D}(\bar x \,\|\, \theta)) = 0. \label{eq:stationary_point_conditions_a}
  \end{align}
Following a similar argument as in \cite{HOFBAUER200747}, $\bar x$ is the Nash equilibrium of the \textit{virtual payoff function} $\tilde{\mathcal{F}}_{\eta, \theta}$, defined as:
\begin{align} \label{eq:virtual_payoff}
  \tilde{\mathcal{F}}_{\eta, \theta}(z) = \mathcal{F}(z) - \eta \nabla \mathcal{D}(z \,\|\, \theta).
\end{align}
The state $\bar x$ is often referred to as the \textit{perturbed} Nash equilibrium of $\mathcal{F}$. Let $\mathbb{PNE}_{\eta, \theta}(\mathcal{F})$ be the set of all perturbed Nash equilibria of $\mathcal{F}$, formally defined as follows.
\begin{definition} \label{def:perturbed_nash}
  Given $\eta > 0$ and $\theta \in \mathrm{int}(\mathbb{X})$, define the set $\mathbb{PNE}_{\eta, \theta}(\mathcal{F})$ of the perturbed Nash equilibria of $\mathcal{F}$ as:
  \begin{align} \label{eq:perturbed_nash}
    \mathbb{PNE}_{\eta, \theta}(\mathcal{F}) = \{ \bar x \in \mathbb{X} \,|\, (\bar x \!-\! z)' \tilde{\mathcal{F}}_{\eta, \theta}(\bar x) \geq 0, \forall z \in \mathbb{X} \},
  \end{align}
  where $\tilde{\mathcal{F}}_{\eta, \theta}$ is the virtual payoff function associated with $\mathcal{F}$.
\end{definition}

\begin{proposition} \label{proposition:convergence_to_PE}
  Consider the closed-loop model \eqref{eq:closed_loop} consisting of the KLD-RL EDM \eqref{eq:closed_loop_a} and PDM \eqref{eq:closed_loop_b}. Suppose that \eqref{eq:closed_loop_b} is weakly $\delta$-antipassive with a deficit of $\bar \nu$, and the weight $\eta$ in \eqref{eq:closed_loop_a} satisfies $\eta > \bar \nu$. Then, the social state $x(t)$ governed by \eqref{eq:closed_loop} converges to the set $\mathbb{PNE}_{\eta, \theta} (\mathcal F)$:
  \begin{align} \label{eq:social_state_convergence}
    \lim_{t \to \infty} \inf_{z \in \mathbb{PNE}_{\eta, \theta} (\mathcal F)} \|x(t) - z\|_2 = 0.
  \end{align}
\end{proposition}

  \begin{proof}
    The proof follows from the technical arguments used in the proof of \cite[Lemma~1]{Park2019Payoff-Dynamic-}, provided that $\eta > \bar \nu$. The original application of \cite[Lemma~1]{Park2019Payoff-Dynamic-} to closed-loop models, where the PDM is defined by \eqref{eq:pdm_finite_dimensional}, also applies here given that \eqref{eq:closed_loop} is well-defined and has a unique solution under Assumption~\ref{assumption:pdm}. \QED
  \end{proof}

Proposition~\ref{proposition:convergence_to_PE} implies that if the surplus of passivity in \eqref{eq:closed_loop_a} exceeds the lack of passivity in \eqref{eq:closed_loop_b},
the social state derived by \eqref{eq:closed_loop} converges to the perturbed Nash equilibrium set \eqref{eq:perturbed_nash}. Consequently, using Lemmas~\ref{lemma:antipassivity_delayed_pdm} and \ref{lemma:antipassivity_smoothing_pdm}, along with Proposition~\ref{proposition:convergence_to_PE}, we can establish convergence to the perturbed Nash equilibrium set for the payoff function with a time-dependent delay \eqref{eq:population_game_time_delay} and the smoothing PDM \eqref{eq:smoothing_pdm}.

\change{
\begin{corollary} \label{corollary:payoff_function_time_delay}
  Suppose the PDM \eqref{eq:closed_loop_b} of the closed-loop model \eqref{eq:closed_loop} is defined by \eqref{eq:population_game_time_delay}. The social state trajectory $x(t), ~ t \geq 0$, determined by \eqref{eq:closed_loop}, converges to $\mathbb{PNE}_{\eta, \theta} (\mathcal F)$ if $\eta > \frac{B_{D \mathcal F}(2+B_{\dot d})}{2}$, where $B_{D \mathcal F}$ is the upper bound on $D\mathcal F$ as defined in Assumption~\ref{assumption:payoff_function} and $B_{\dot d}$ is defined such that $\dot d(t) \in [-B_{\dot d}, 1)$ holds.
\end{corollary}
}

\begin{corollary} \label{corollary:smoothing_pdm}
  Suppose the PDM \eqref{eq:closed_loop_b} of the closed-loop model \eqref{eq:closed_loop} is defined by \eqref{eq:smoothing_pdm}, for which the underlying game is contractive and has an affine payoff function $\mathcal F(x) = Fx + b$. The social state trajectory $x(t), ~ t \geq 0$, determined by \eqref{eq:closed_loop}, converges to $\mathbb{PNE}_{\eta, \theta} (\mathcal F)$ if $\eta > \frac{1}{4} \left\| F - F' \right\|_2$.
\end{corollary}

\subsection{Iterative KLD Regularization and Convergence Guarantee} \label{sec:iterative_regularization}
Suppose the population game $\mathcal F$ underlying the PDM \eqref{eq:pdm} has a Nash equilibrium $x^{\tiny \text{NE}}$ in $\mathrm{int}(\mathbb X)$. If $\theta$ coincides with $x^{\tiny \text{NE}}$, then $\mathbb{PNE}_{\eta,\theta} (\mathcal F) = \{ x^{\tiny \text{NE}} \}$ and under the conditions of Proposition~\ref{proposition:convergence_to_PE}, the social state $x(t)$ converges to $x^{\tiny \text{NE}}$. Therefore, achieving convergence to the Nash equilibrium set requires $\theta = x^{\tiny \text{NE}}$. In this section, we discuss the design of a parameter update algorithm that specifies how the agents update $\theta$ to asymptotically attain the Nash equilibrium.

Let the social state $x(t)$ evolve according to \eqref{eq:closed_loop} and the regularization parameter $\theta$ be iteratively updated at each time instant of a discrete-time sequence $\{t_l\}_{l=1}^\infty$ as $\theta = x(t_l) \in \mathrm{int}(\mathbb X)$. That is, $\theta$ is reset to the current social state $x(t_l)$ at each time $t_l$ and maintains its value until the next update time $t_{l+1}$.
Let $\{\theta_l\}_{l=1}^\infty$ be the resulting sequence of parameter updates, i.e., $\theta_l = x(t_l)$.
Suppose the following two conditions hold:
\begin{subequations} \label{eq:parameter_update_conditions}
  \begin{align}
    &\max_{z \in \mathbb X} ( z \!-\! \theta_{l+1} )' ( \mathcal F(\theta_{l+1})
      \!-\! \eta \nabla \mathcal D ( \theta_{l+1} \| \theta_l ) ) \leq \frac{\eta}{2} \mathcal D ( \theta_{l+1} \| \theta_l )  \label{eq:parameter_update_conditions_a} \\
    &\lim_{l \to \infty} \mathcal D ( \theta_{l+1} \| \theta_l ) = 0 \nonumber  \\
    &\implies \!\lim_{l \to \infty} \!\alpha_{x,p} (t_l) \!=\! 0 \text{ and } \lim_{l \to \infty} \!\| p(t_l) \!-\! \mathcal F(\theta_l) \|_2 \!=\! 0, \label{eq:parameter_update_conditions_b}
  \end{align}
\end{subequations}
where $\alpha_{x,p}$ is the function defined in \eqref{eq:weakAntipassivity} for the PDM \eqref{eq:closed_loop_b}. According to \eqref{eq:stationary_point_conditions_a} and \eqref{eq:perturbed_nash}, 
the condition \eqref{eq:parameter_update_conditions_a} means that $\theta$ is updated to a new value, i.e., $\theta_{l+1} = x(t_{l+1})$, when the state $x(t_{l+1})$ is sufficiently close to $\mathbb{PNE}_{\eta, \theta_l} (\mathcal F)$. The  condition \eqref{eq:parameter_update_conditions_b} implies that as the sequence $\{\theta_l\}_{l=1}^\infty$ converges, the estimated stored energy in the PDM \eqref{eq:closed_loop_b} dissipates and the payoff vector $p(t_l)$ converges to $\mathcal F(\theta_l)$.

The following lemma states the convergence of the social state $x(t)$ to the Nash equilibrium set $\mathbb{NE}(\mathcal F)$ if the sequences $\{ t_l \}_{l=1}^\infty$ and $\{ \theta_l \}_{l=1}^\infty$ satisfy \eqref{eq:parameter_update_conditions}. The proof of the lemma is given in Appendix~\ref{proof_lemma:stability}. 

\begin{lemma} \label{lemma:stability}
  Consider that the social state $x(t)$ and payoff vector $p(t)$ evolve according to  the closed-loop model \eqref{eq:closed_loop}, 
  the PDM \eqref{eq:closed_loop_b} is weakly $\delta$-antipassive with deficit $\bar \nu$, and the weight $\eta$ of the KLD-RL EDM \eqref{eq:closed_loop_a} satisfies $\eta > \bar \nu$. Suppose the parameter $\theta$ of \eqref{eq:closed_loop_a} is iteratively updated according to $\theta_l = x(t_l), \,l \in \mathbb N$ such that \eqref{eq:parameter_update_conditions} holds, and one of the following two conditions holds.\footnote{The condition \ref{stability_contractive_games}) implies that Nash equilibria of a contractive game $\mathcal F$ can be located at any location in $\mathbb X$ as long as at least one of them belongs to $\mathrm{int}(\mathbb X)$.
  }
  \begin{enumerate}
    \renewcommand{\theenumi}{(C\arabic{enumi}}
  \item $\mathcal F$ is contractive and has a Nash equilibrium in $\mathrm{int}(\mathbb X)$. \label{stability_contractive_games}
  \item $\mathcal F$ is strictly contractive. \label{stability_strictly_contractive_games}
  \end{enumerate}
  Then, the state $x(t)$ converges to the Nash equilibrium set:
  \begin{align}
    \lim_{t \to \infty} \inf_{z \in \mathbb{NE} (\mathcal F)} \|x(t) - z\|_2 = 0.
  \end{align}
\end{lemma}

Lemma~\ref{lemma:stability} suggests that if the parameter $\theta$ is updated in a way that the resulting sequence of the parameter updates satisfies \eqref{eq:parameter_update_conditions}, the convergence to the Nash equilibrium set is guaranteed. Unlike in Proposition~\ref{proposition:convergence_to_PE}, the underlying population game $\mathcal F$ needs to be contractive to establish the convergence.

To evaluate \eqref{eq:parameter_update_conditions_a} for the parameter update, since the agents may not have access to the quantity $\mathcal F(x(t))$, they need to estimate it using the payoff vector $p(t)$. Suppose the estimation error $\|p(t) - \mathcal F (x(t))\|_2$ is bounded by a function $\beta_{x,p}(t)$, i.e.,
\begin{align}
  \|p(t) - \mathcal F (x(t))\|_2 \leq \beta_{x,p}(t), \, \forall t \geq 0,
\end{align}
for which $\lim_{t \to \infty} \beta_{x,p}(t) = 0$ holds whenever the trajectory $x(t), \, t \geq 0$ satisfies $\lim_{t \to \infty} \|\dot x(t)\|_2 = 0$. Note that according to Assumption~\ref{assumption:pdm}-1 such a function $\beta_{x,p}$ always exists. Thus, we derive the following relation.
\begin{align} \label{eq:inequality_for_algorithm}
  &\max_{z \in \mathbb X} \left( z - \theta_{l+1} \right)' \left( \mathcal F(\theta_{l+1}) - \eta \nabla \mathcal D \left( \theta_{l+1} \,\|\, \theta_{l} \right) \right) \nonumber \\
  &\leq \max_{z \in \mathbb X} \left( z - \theta_{l+1} \right)' \left( p (t_{l+1}) - \eta \nabla \mathcal D \left( \theta_{l+1} \,\|\, \theta_{l} \right) \right) \nonumber \\
  &\qquad \qquad + \max_{z \in \mathbb X} \left( z - \theta_{l+1} \right)' \left( \mathcal F(\theta_{l+1}) - p (t_{l+1}) \right) \nonumber \\
  &\leq \max_{z \in \mathbb X} \left( z - \theta_{l+1} \right)' \left( p(t_{l+1}) - \eta \nabla \mathcal D \left( \theta_{l+1} \,\|\, \theta_{l} \right) \right) \nonumber \\
  &\qquad \qquad + \sqrt{2M} \beta_{x,p}(t_{l+1}).
\end{align}
Using \eqref{eq:inequality_for_algorithm}, we can verify that \eqref{eq:parameter_update_conditions} holds if the parameter is updated as $\theta_{l+1} = x(t_{l+1})$ at each $t_{l+1}$ satisfying
\begin{multline} \label{eq:parameter_update_rule}
  \max_{z \in \mathbb X} ( z - \theta_{l+1} )' ( p(t_{l+1}) - \eta \nabla \mathcal D ( \theta_{l+1} \| \theta_l ) ) \\ + \alpha_{x,p}(t_{l+1}) + \sqrt{2M} \beta_{x,p}(t_{l+1})
  \leq \frac{\eta}{2} \mathcal D ( \theta_{l+1} \| \theta_l ).
\end{multline}

In what follows, we discuss whether such time instant $t_{l+1}$ always exists. Suppose the parameter $\theta$ is fixed to $\theta_l = x(t_l)$. According to Proposition~\ref{proposition:convergence_to_PE} and the definition of the KLD-RL EDM \eqref{eq:closed_loop_a}, the social state $x(t)$ converges to the perturbed Nash equilibrium set $\mathbb{PNE}_{\eta, \theta_l}(\mathcal F)$ which implies
\begin{subequations} \label{eq:convergence_to_PE_implications}
  \begin{align}
    &\lim_{t \to \infty} \| \dot x(t) \|_2 \to 0 \\
    &\lim_{t \to \infty} \textstyle \max_{z \in \mathbb X} ( z \!-\! x(t) )' ( p(t) \!-\! \eta \nabla \mathcal D ( x(t) \| \theta_l ) ) \to 0.
  \end{align}
\end{subequations}
Recall that $\lim_{t \to \infty} \|\dot x(t) \|_2 = 0$ implies $\lim_{t \to \infty} \alpha_{x,p}(t) = 0$ and $\lim_{t \to \infty} \beta_{x,p}(t) = 0$.
Hence, by \eqref{eq:convergence_to_PE_implications}, the following term 
vanishes as $t$ tends to infinity:
\begin{multline*}
  \max_{z \in \mathbb X} ( z - x(t) )' ( p(t) - \eta \nabla \mathcal D ( x(t) \| \theta_l ) ) \\ + \alpha_{x,p}(t) + \sqrt{2M} \beta_{x,p}(t).
\end{multline*}
Consequently, either we can find $t_{l+1}$ satisfying \eqref{eq:parameter_update_rule} or the state $x(t)$ converges to $\theta_l$, i.e., $\lim_{t \to \infty} \mathcal D(x(t) \,\|\, \theta_l) = 0$. However, by Proposition~\ref{proposition:convergence_to_PE} and Definition~\ref{def:perturbed_nash}, the latter case implies that $\theta_l$ must be the Nash equilibrium and, hence, the limit point of $x(t)$. In conclusion, for both cases, by applying Lemma~\ref{lemma:stability}, the convergence of the social state $x(t)$ to the Nash equilibrium set is guaranteed. In what follows, we only consider the case where the parameter update rule \eqref{eq:parameter_update_rule} yields an infinite sequence.

Inspired by \eqref{eq:parameter_update_rule}, we propose Algorithm~\ref{algorithm:parameter_update} to realize such parameter update for cases where the PDM \eqref{eq:closed_loop_b} is \eqref{eq:population_game_time_delay} or \eqref{eq:smoothing_pdm}.

\begin{algorithm}
  \caption{At each time $t$, \texttt{Parameter\_Update} is executed to check for parameter updates.}
  \label{algorithm:parameter_update}
  \DontPrintSemicolon
  \SetKwFunction{FCheckCondition}{Check\_Condition}
  \SetKwProg{Fn}{Function}{:}{}
  \SetKwFunction{proc}{Parameter\_Update}
  \SetKwProg{myproc}{Procedure}{}{}
  
  \Fn{\FCheckCondition{$t$}}{
    Check Eq. \eqref{eq:bias_selection_criteria_1} (or Eq. \eqref{eq:bias_selection_criteria_2}) using $x(t)$, $p(t)$, $\theta$, then return the result ($True/False$).
  } 
  \myproc{\proc{t}}{
    \If{$\FCheckCondition{$t$} == True$} {
      $\theta \gets x(t)$ {\small \tcp*[l]{update parameter $\theta$}}
    }
  }
\end{algorithm}

\begin{enumerate}
\item \textbf{Payoff function with a time-dependent delay \eqref{eq:population_game_time_delay}:}
  \begin{align}
    \label{eq:bias_selection_criteria_1}
    &\max_{z \in \mathbb X} ( z - x(t) )' ( p(t) - \eta \nabla \mathcal D ( x(t) \,\|\, \theta ) ) \nonumber \\
    &\qquad + \sqrt{2M} B_{D\mathcal F} B_d \max_{\tau \in [t - B_d, t]} \|\dot x (\tau)\|_2 \nonumber \\
    &\leq \frac{\eta}{2} \mathcal D ( x (t) \,\|\, \theta ).
  \end{align}
  
\item \textbf{Smoothing PDM \eqref{eq:smoothing_pdm}:}
  \begin{align}
    \label{eq:bias_selection_criteria_2}
    &\max_{z \in \mathbb X} ( z - x(t) )' ( p(t) - \eta \nabla \mathcal D ( x(t) \,\|\, \theta ) ) \nonumber \\
    & \qquad + \sqrt{2M} \Big( \left( \left\| p(\gamma t) \right\|_2 + B_{\mathcal F} \right) \exp({-\lambda \left(1 - \gamma \right) t }) \nonumber \\
    & \qquad + B_{D \mathcal F} \int_{\gamma t}^{t} \exp({-\lambda (t - \tau)}) \| \dot x(\tau) \|_2 \, \mathrm d\tau \Big) \nonumber \\
    & \leq \frac{\eta}{2} \mathcal D ( x (t) \,\|\, \theta ).
  \end{align}
\end{enumerate}
where $\gamma$ is any fixed real number in $(0, 1)$ and $\dot x(\tau)$ is computed using \eqref{eq:closed_loop_a}. 

To realize Algorithm~\ref{algorithm:parameter_update}, the agents only need to know the upper bounds $B_d, B_{D \mathcal F}$ on $d, D \mathcal F$ for \eqref{eq:population_game_time_delay}, or the bounds $B_{\mathcal F}, B_{D \mathcal F}$ on $\mathcal F, D \mathcal F$ for \eqref{eq:smoothing_pdm}. 
The motivation behind adopting the algorithm is analogous to iterative regularization techniques that have been frequently used in optimization and machine learning methods to avoid issues with over-fitting or fluctuation, especially when cost functions are noisy \cite{10.5555/1734076}. Similarly, in our multi-agent learning problem, Algorithm~\ref{algorithm:parameter_update} prevents the state trajectory from exhibiting oscillation and ensures the convergence to the Nash equilibrium set.

\change{
Recall the conditions \ref{stability_contractive_games}) and \ref{stability_strictly_contractive_games}) stated in Lemma~\ref{lemma:stability} as follows:
\begin{enumerate}
    \renewcommand{\theenumi}{(C\arabic{enumi}}
  \item $\mathcal F$ is contractive and has a Nash equilibrium in $\mathrm{int}(\mathbb X)$. \label{stability_contractive_games}
  \item $\mathcal F$ is strictly contractive. \label{stability_strictly_contractive_games}
  \end{enumerate}
Leveraging Lemmas~\ref{lemma:antipassivity_delayed_pdm}-\ref{lemma:stability}, in the following theorems, we state convergence results for the closed-loop model  \eqref{eq:closed_loop} when $\theta$ of \eqref{eq:closed_loop_a} is updated according to Algorithm~\ref{algorithm:parameter_update} and \eqref{eq:closed_loop_b} is a payoff function with a time-dependent delay \eqref{eq:population_game_time_delay} (Theorem~\ref{theorem:stability_time_delay}) or smoothing PDM \eqref{eq:smoothing_pdm} (Theorem~\ref{theorem:stability_smoothing_pdm}). The proofs of the theorems are provided in Appendix~\ref{proof_theorem:stability}.
\begin{theorem} \label{theorem:stability_time_delay}
  Suppose that the social state $x(t)$ and payoff vector $p(t)$ are derived by the closed-loop model \eqref{eq:closed_loop}, where the PDM \eqref{eq:closed_loop_b} is a payoff function with a time-dependent delay \eqref{eq:population_game_time_delay}.
  If $\mathcal F$ satisfies one of the conditions \ref{stability_contractive_games}) or \ref{stability_strictly_contractive_games}), $\eta > \frac{B_{D \mathcal F}(2+B_{\dot d})}{2}$ holds, and the parameter $\theta$ is iteratively updated according to Algorithm~\ref{algorithm:parameter_update} with  \eqref{eq:bias_selection_criteria_1}, where $B_{D \mathcal F}$ is the upper bound on $D\mathcal F$ as defined in Assumption~\ref{assumption:payoff_function} and $B_{\dot d}$ is defined such that $\dot d(t) \in [-B_{\dot d}, 1)$ holds, then the state $x(t)$ converges to the Nash equilibrium set.
\end{theorem}}

\begin{theorem} \label{theorem:stability_smoothing_pdm}
  Suppose that the social state $x(t)$ and payoff vector $p(t)$ are derived by the closed-loop model \eqref{eq:closed_loop}, where the PDM \eqref{eq:closed_loop_b} is a smoothing PDM \eqref{eq:smoothing_pdm} with $\mathcal F(x) = Fx + b$. If $\mathcal F$ satisfies one of the conditions \ref{stability_contractive_games}) or \ref{stability_strictly_contractive_games}),
  $\eta > \frac{1}{4} \|F - F' \|_2$ holds, and the parameter $\theta$ is iteratively updated according to Algorithm~\ref{algorithm:parameter_update} with \eqref{eq:bias_selection_criteria_2}, then the state $x(t)$ converges to the Nash equilibrium set.
\end{theorem}

  \section{Distributed Parameter Update} \label{section:distributed_parameter_update}
  In this section, we consider the scenario where the state and payoff vector associated with a population are private information, and agents within each population do not want to share it with others outside the population. Exposing such private information could compromise their decision making; for example, using the information, agents from an adversarial population could revise their strategies to deliberately diminish the payoff vector of the exposed population. However, it is important to note that to validate the conditions \eqref{eq:bias_selection_criteria_1} or \eqref{eq:bias_selection_criteria_2} of Algorithm~\ref{algorithm:parameter_update}, every population in the society needs to assess the states and payoff vector of all other populations. In what follows, we discuss how Algorithm~\ref{algorithm:parameter_update} can be revised to address the privacy issue. We also explain how the parameter update can be executed in a distributed manner by individual agents.

  We proceed with restating the conditions \eqref{eq:bias_selection_criteria_1} and \eqref{eq:bias_selection_criteria_2} for each population~$k$, respectively, as follows:
  \begin{enumerate}
  \item \textbf{Payoff function with a time-dependent delay \eqref{eq:population_game_time_delay}:}
    \begin{align}
      \label{eq:bias_selection_criteria_1_distributed}
      &\max_{z^k \in \mathbb X^k} ( z^k - x^k(t) )' ( p^k(t) - \eta \nabla \mathcal D ( x^k(t) \| \theta^k ) ) \nonumber \\
      &\qquad +\sqrt{2M} B_{D\mathcal F} B_d \max_{\tau \in [t - B_d, t]}  \| \dot x^k(\tau) \|_2 \nonumber \\
      &\leq \frac{\eta}{2} \mathcal D ( x^k (t) \| \theta^k ).
    \end{align}

  \item \textbf{Smoothing PDM \eqref{eq:smoothing_pdm}:}
    \begin{align}
      \label{eq:bias_selection_criteria_2_distributed}
      &\max_{z^k \in \mathbb X^k} ( z^k - x^k(t) )' ( p^k(t) - \eta \nabla \mathcal D ( x^k(t) \,\|\, \theta^k ) ) \nonumber \\
      & \quad + \sqrt{2M} \Big( \Big( \| p^k(\gamma t) \|_2 + \frac{B_{\mathcal F}}{M} \Big) \exp (-\lambda (1 - \gamma ) t ) \nonumber \\
      & \quad + B_{D \mathcal F} \int_{\gamma t}^{t} \exp(-\lambda (t - \tau)) \| \dot x^k(\tau) \|_2 \, \mathrm d\tau \Big) \nonumber \\
      & \leq \frac{\eta}{2} \mathcal D ( x^k (t) \,\|\, \theta^k ).
    \end{align}
  \end{enumerate}
  Using the facts that $\| \dot x(\tau) \|_2 \leq \sum_{k=1}^M  \| \dot x^k(\tau) \|_2$ and $\| p(\gamma t) \|_2 \leq \sum_{k=1}^M  \|  p^k(\gamma t) \|_2$, we can validate \eqref{eq:bias_selection_criteria_1} (or \eqref{eq:bias_selection_criteria_2}) if \eqref{eq:bias_selection_criteria_1_distributed} (or \eqref{eq:bias_selection_criteria_2_distributed}) holds for all $k$ in $\{1, \cdots, M\}$.
  Consequently, agents in population~$k$ only need to inform agents in all other populations whether the condition \eqref{eq:bias_selection_criteria_1_distributed} (or \eqref{eq:bias_selection_criteria_2_distributed}) is satisfied at each time $t$, without sharing other information including their parameter $\theta^k$, population state $x^k(t)$, and payoff vector $p^k(t)$.

  In Algorithm~\ref{algorithm:distributed_parameter_update}, we present a distributed algorithm in which an individual agent in each population~$k$ maintains its own parameter $\theta^k$ and decides when to update its value.\footnote{Since Algorithm~\ref{algorithm:distributed_parameter_update} ensures that every agent in each population~$k$ maintains the same value of $\theta^k$, without loss of generality, we use the same notation $\theta^k$ to denote the regularization parameter for any agent in the population.} Through communication with their respective neighbors across the society, all agents reach an agreement on a single time instant at which the condition \eqref{eq:bias_selection_criteria_1_distributed} (or \eqref{eq:bias_selection_criteria_2_distributed}) is satisfied for every population. Consequently, the algorithm facilitates a distributed update of the parameters $\theta^1, \cdots, \theta^M$ without requiring a centralized coordinator.

  In the algorithm, $G = (\mathbb V, \mathbb E)$ represents a communication graph, where $\mathbb V$ is the set of nodes representing all agents in the society, and $\mathbb E$ represents directed edges indicating the viability and direction of communication between two agents. We consider the scenario where the agents communicate with their neighbors at intervals of $T_{\text{comm}}$; that is, they can share information at $t = T_{\text{comm}}, \, 2 T_{\text{comm}}, \, 3 T_{\text{comm}}, \, \cdots$.  
  According to the procedure {\tt Distributed\_Parameter\_Update($t$)}, using the time variable $\tau_{\text{update}}$ and the boolean variable $b_{\text{update}}$, each agent in population~$k$ tracks the earliest time $t$ at which the condition \eqref{eq:bias_selection_criteria_1_distributed} (or \eqref{eq:bias_selection_criteria_2_distributed}) is satisfied using its own information on $x^k(t), p^k(t), \theta^k$ (Lines 4 -- 6). When a neighbor $l$ shares its recorded $\tau_{\text{update}}^{(l)}$, the agent compares it with its own $\tau_{\text{update}}$. If $\tau_{\text{update}}^{(l)}$ is larger, then the agent finds the earliest time $\tau$ in $[\tau_{\text{update}}^{(l)}, t]$ at which \eqref{eq:bias_selection_criteria_1_distributed} (or \eqref{eq:bias_selection_criteria_2_distributed}) is satisfied (Lines 10 -- 12). If such $\tau$ exists then set $\tau_{\text{update}} = \tau$ and $b_{\text{update}} = True$; otherwise, the agent resets both $\tau_{\text{update}}$ and $b_{\text{update}}$ to the current time $t$ and $False$, respectively (Lines 13 -- 18).

  If there is no update of $\tau_{\text{update}}$ for longer than the period of $2\,\mathrm{diameter}(G) \times T_{\text{comm}}$, where $\mathrm{diameter}(G)$ is the diameter of the graph $G$, then the agent updates its parameter $\theta^{k}$ with $x^k(\tau_{\text{update}})$ (Lines 19 -- 21).
  When $G$ is strongly connected, under Algorithm~\ref{algorithm:distributed_parameter_update}, the condition $\tau_{\text{update}} + 2\,\mathrm{diameter}(G) \times T_{\text{comm}} > t$ in Line 19 implies that every agent in the society reaches consensus on the same value of $\tau_{\text{update}}$ and $b_{\text{update}} = True$ holds, indicating that  \eqref{eq:bias_selection_criteria_1_distributed} (or \eqref{eq:bias_selection_criteria_2_distributed}) is satisfied for every $k$ in $\{1, \cdots, M\}$. Hence, the algorithm guarantees synchronized parameter updates through communication within each agent's neighborhood. However, since the agents must wait a period of $2\,\mathrm{diameter}(G) \times T_{\text{comm}}$ from the time $\tau_{\text{update}}$ at which the condition \eqref{eq:bias_selection_criteria_1_distributed} (or \eqref{eq:bias_selection_criteria_2_distributed}) is met for every $k$ before they update their parameter $\theta^k$, this period can be considered the lag in the distributed parameter update.

\begin{algorithm}
    \caption{At each time $t$, an individual agent in population~$k$ executes \texttt{Distributed\_Parameter\_Update} to check for parameter updates.}
    \label{algorithm:distributed_parameter_update}
    \DontPrintSemicolon
    \SetKwFunction{FCheckCondition}{Check\_Condition}
    \SetKwProg{Fn}{Function}{:}{}
    \SetKwFunction{proc}{Distributed\_Parameter\_Update}
    \SetKwProg{myproc}{Procedure}{}{}
    
    \Fn{\FCheckCondition{$t$}}{
      Check Eq. \eqref{eq:bias_selection_criteria_1_distributed} (or Eq. \eqref{eq:bias_selection_criteria_2_distributed}) using $x^k(t)$, $p^k(t)$, $\theta^{k}$, then return the result ($True/False$).
    }
    \myproc{\proc{t}}{
    \If{$b_{\text{update}} == False$} {
        $\tau_{\text{update}} \gets t$ \;
        $b_{\text{update}} \gets \FCheckCondition{$t$}$ \;
        }
        
      \If{$t \in \{T_{\text{comm}}, 2 T_{\text{comm}}, 3 T_{\text{comm}}, \cdots \}$} {
        Share $\tau_{\text{update}}$ with out-neighbors \;
        \For{$l \in \mathbb N_{\text{in}}$} { {\footnotesize\tcp*[l]{process data $\tau_{\text{update}}^{(l)}$ from each agent~$l$ in the in-neighborhood $\mathbb N_{\text{in}}$.}}
          \If{$\tau_{\text{update}} < \tau_{\text{update}}^{(l)}$} {
            $\mathbb T \!\gets\! \{ \tau_{\text{update}}^{(l)} \!\leq \!\tau \!\leq\! t \big|$\; $~~~\FCheckCondition{$\tau$} \!==\! True \}$ \;
            \uIf{$\mathbb T$ is not empty} {
              $\tau_{\text{update}} \gets \min \mathbb T$ \;
              $b_{\text{update}} \gets True$ \;
            }
            \Else {            
              $\tau_{\text{update}} \gets t$ \;
              $b_{\text{update}} \gets False$ \;
            }
          }
        }
      }
      \If{$\tau_{\text{update}} + 2\,\mathrm{diameter}(G) \times T_{\text{comm}} > t$}{
        $b_{\text{update}} \gets False$ \;
        $\theta^{k} \gets x^k(\tau_{\text{update}})$
        {\tcp*[l]{\small update parameter $\theta^k$}}
      }
      
    }
\end{algorithm}

\section{Simulations with Numerical Examples} \label{sec:simulation}

\subsection{Convergence and Performance Improvements} \label{sec:simulation_a}
\begin{figure} [t]
  \center
  \subfigure[]{
    \includegraphics[trim={.0in .2in 0 .1in},width=1.6in]{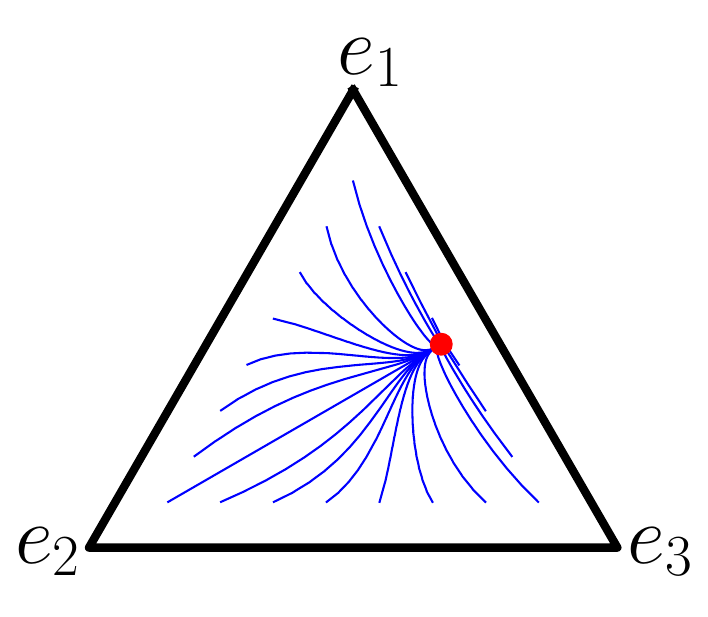}
    \label{fig:kld_rl_a}
  }
  \subfigure[]{
    \includegraphics[trim={.0in .2in 0 .1in},width=1.6in]{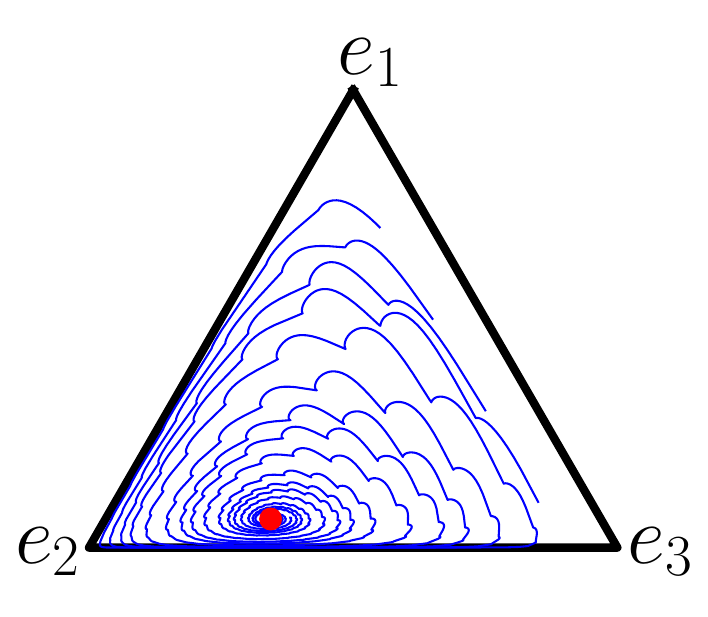}
    \label{fig:kld_rl_b}
  }
  \caption{State trajectories of population~1 derived by the closed-loop model \eqref{eq:closed_loop}. The PDM \eqref{eq:closed_loop_b} is defined by (a) the payoff function with a fixed unit time delay with $\mathcal F$ given by \eqref{eq:congestion_game} and (b) the smoothing PDM with $\mathcal F$ defined by \eqref{eq:rps_game}. The parameter $\eta$ of the EDM \eqref{eq:closed_loop_a} is selected as $\eta=4.5$ for (a) and $\eta=0.6$ for (b).  The red circle in (a) and (b) marks the Nash equilibrium of each game.}
  \label{fig:kld_rl}
\end{figure}

\begin{figure} [t]
  \center
  \subfigure[] {\includegraphics[trim={.1in .2in 0 0},width=1.65in]{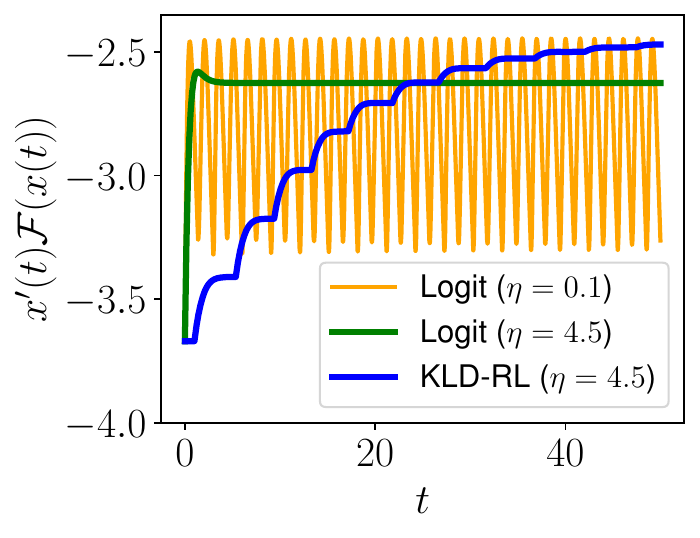}
    \label{fig:performance_comparison_a}
  }
  \subfigure[] {\includegraphics[trim={.1in .2in 0 0},width=1.65in]{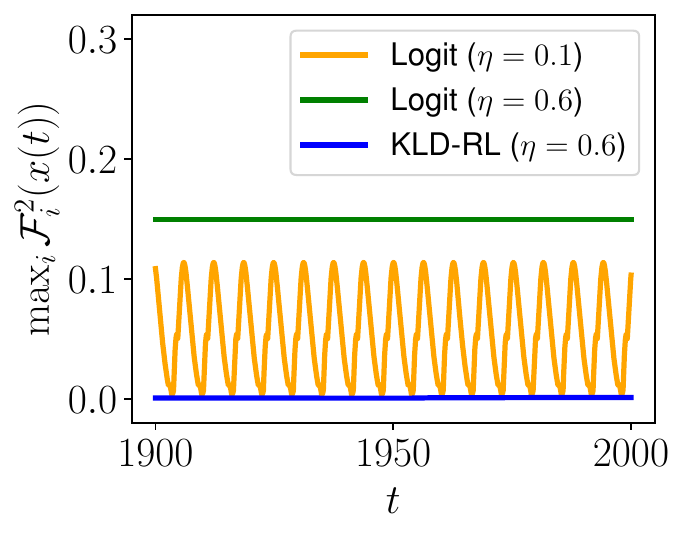}
    \label{fig:performance_comparison_b}
  }

  \caption{(a) Average payoff $x'(t) \mathcal F (x(t))$ of the entire society in the congestion game and (b) maximum gain $\max_{i \in \{1, 2, 3\}} \mathcal F_i^2(x(t))$ of population~2 (the opponent of population~1) in the zero-sum game for the logit protocol and KLD-RL protocol.}
  \label{fig:performance_comparison}
  \vspace{-1ex}
\end{figure}

We use Examples~\ref{example:congestion_game} and \ref{example:zero_sum_game} to illustrate our main results. We evaluate the social state trajectory of the closed-loop model \eqref{eq:closed_loop} when the PDM \eqref{eq:closed_loop_b} is defined by \textbf{(i)} the payoff function with a time-dependent delay \eqref{eq:population_game_time_delay}, with $\mathcal F$ defined as in \eqref{eq:congestion_game}, and \textbf{(ii)} the smoothing PDM \eqref{eq:smoothing_pdm}, with $\mathcal F$ defined as in \eqref{eq:rps_game}.

\subsubsection{Congestion Population Game with Time Delay} \label{sec:congestion_game_simulation}
We illustrate our main results using the congestion population game \eqref{eq:congestion_game} with a fixed unit time delay ($d(t)=1,~\forall t \geq 0$); hence, $B_{\dot d} = 0$. For \eqref{eq:closed_loop_a}, we set $\eta = 4.5$ to ensure that $\eta > B_{D \mathcal F}$, where $B_{D \mathcal F}$ is chosen to be the $2$-norm of the payoff matrix \eqref{eq:congestion_game}, and update the parameter $\theta$ using Algorithm~\ref{algorithm:parameter_update} with \eqref{eq:bias_selection_criteria_1}. For simplicity, we assign $B_d = 1$.

Fig.~\ref{fig:kld_rl_a} illustrates that resulting social state trajectories converge to the unique Nash equilibrium of \eqref{eq:congestion_game}. Recall that for the logit protocol case, as depicted in Fig.~\ref{fig:standard_logit_models_delayed_payoff}, the social state trajectories either exhibit oscillation around the Nash equilibrium (when $\eta$ is small) or converge to a stationary point located away from the Nash equilibrium (when $\eta$ is large).

  To assess the efficacy of the KLD-RL model in enabling agents to learn an effective strategy profile in the congestion game, we compare the average payoff  $x'(t) \mathcal F (x(t))$ using the social state $x(t)$ derived by the logit protocol \eqref{eq:StandardLogitProtocol} with $\eta = 0.1, 4.5$ and the KLD-RL protocol \eqref{eq:kld_rl} with ${\eta = 4.5}$. As shown in Fig.~\ref{fig:performance_comparison_a}, the social state trajectory determined by \eqref{eq:kld_rl} converges to the strategy profile attaining the maximum average payoff, approximately $-2.44$, whereas the trajectories determined by the logit protocol do not.

\subsubsection{Zero-sum Game with Smoothing PDM} \label{sec:zerosum_game_simulation}
We iterate the simulations using the smoothing PDM \eqref{eq:smoothing_pdm} defined by the zero-sum game \eqref{eq:rps_game}. We assign $\lambda = 1$ for \eqref{eq:smoothing_pdm}, set $\eta = 0.6$ for \eqref{eq:closed_loop_a} to ensure that $\eta > \frac{1}{4} \|F - F' \|_2$ holds, where $F$ is the payoff matrix derived from \eqref{eq:rps_game}, and update the parameter $\theta$ using Algorithm~\ref{algorithm:parameter_update} with \eqref{eq:bias_selection_criteria_2}, where $\gamma = 0.1$.\footnote{We select $\gamma = 0.1$ as it provides fast convergence of the social state to the Nash equilibrium, which we validated through multiple rounds of simulations using different values of $\gamma$. Recall that, as stated in Theorem~\ref{theorem:stability_smoothing_pdm}, for any $\gamma$ in $(0,1)$, the social state converges to the Nash equilibrium set.}

As illustrated in Fig.~\ref{fig:kld_rl_b}, the resulting social state trajectories converge to the unique Nash equilibrium of \eqref{eq:rps_game}. Similar to the time-delay case considered in \S\ref{sec:congestion_game_simulation}, as shown in Fig.~\ref{fig:standard_logit_models_smoothing_pdm}, the state trajectories derived by the logit protocol either exhibit oscillation around the Nash equilibrium (when $\eta$ is small) or converge to a stationary point located away from the Nash equilibrium (when $\eta$ is large).

  To assess the efficacy of the KLD-RL model in enabling each population to learn an effective strategy profile in the zero-sum game, we evaluate the maximum gain $\max_{i \in \{1, 2, 3\}} \mathcal F_i^2 (x(t))$ of population~2 (the opponent of population~1), using the social state $x(t)$ resulting from the logit protocol \eqref{eq:StandardLogitProtocol} with $\eta = 0.1, 0.6$ and the KLD-RL protocol \eqref{eq:kld_rl} with ${\eta = 0.6}$. As shown in Fig.~\ref{fig:performance_comparison_b}, the social state trajectory determined by \eqref{eq:kld_rl} converges to the strategy profile at which the maximum gain of population~2 is minimized, whereas the trajectories determined by the logit protocol do not.

\subsection{Selection of Regularization Weight $\eta$}
\begin{figure} [t]
  \center
  \subfigure[]{
    \includegraphics[trim={.1in .2in 0 .1in},width=1.61in]{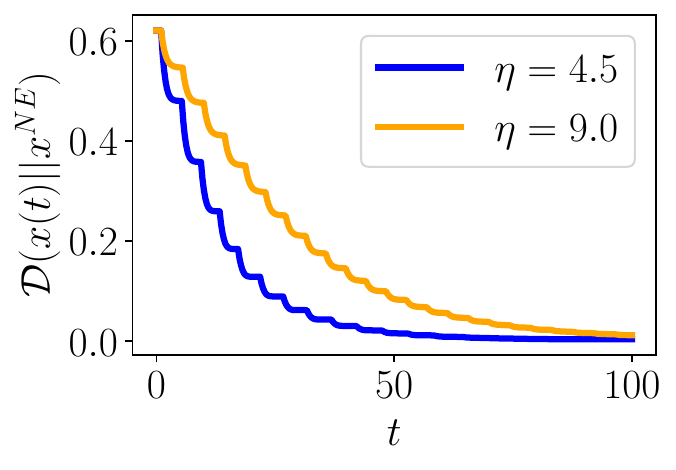}
  }
  \subfigure[]{
    \includegraphics[trim={.1in .2in 0 .1in},width=1.61in]{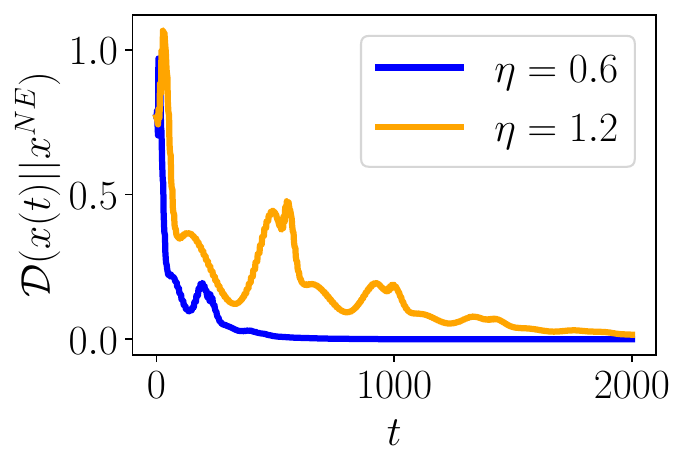}
  }
  \caption{Graphs of the KL divergence $\mathcal D(x(t) \| x^{\text{\scriptsize NE}} )$ between the social state and the Nash equilibrium in the (a) congestion game and (b) zero-sum game.
  }
  \label{fig:kl_learning_convergence_speed}
\end{figure}

As  stated in Theorems~\ref{theorem:stability_time_delay} and \ref{theorem:stability_smoothing_pdm}, a proper selection of $\eta$ in \eqref{eq:closed_loop_a} is critical. In particular, for \eqref{eq:population_game_time_delay} and \eqref{eq:smoothing_pdm}, $\eta$ needs to be larger than $\frac{B_{D \mathcal F}(2+B_{\dot d})}{2}$ and $\frac{1}{4} \|F - F'\|_2$, respectively, to ensure  convergence to the Nash equilibrium set. In practice, unless the payoff function $\mathcal F$ is available, the agents need to estimate upper bounds on $\frac{B_{D \mathcal F}(2+B_{\dot d})}{2}$ and $\frac{1}{4} \|F - F'\|_2$ and select $\eta$ to be larger than these estimates. When these estimates are conservative, the agents will select unnecessarily large $\eta$.

Recall that $\eta$ determines the weight on the regularization in the KLD-RL protocol \eqref{eq:maximization}: Smaller $\eta$ means that the regularization has little effect on the strategy revision and the agents tend to select strategies that increase the average payoff. Hence, assigning smaller $\eta$ would imply faster convergence to the Nash equilibrium set. We verify such observation using simulations.
Fig.~\ref{fig:kl_learning_convergence_speed} depicts the simulation results with two different values of $\eta$ in each of the two PDM cases we considered in \S\ref{sec:simulation_a}. In both cases, when $\eta$ is small, we can observe that the social state converges to the unique Nash equilibrium of each game faster than when $\eta$ is large. In conclusion, for achieving faster convergence to the Nash equilibrium set, the agents need to assign a small value to $\eta$ satisfying the inequality requirement $\eta > \bar \nu$ for the convergence as we discussed in Lemma~\ref{lemma:stability}, and this will require computing an accurate estimate of $\bar \nu$.

\begin{figure} [t]
  \center
  \subfigure[]{
    \includegraphics[trim={.0in .2in 0 .1in},width=1.6in]{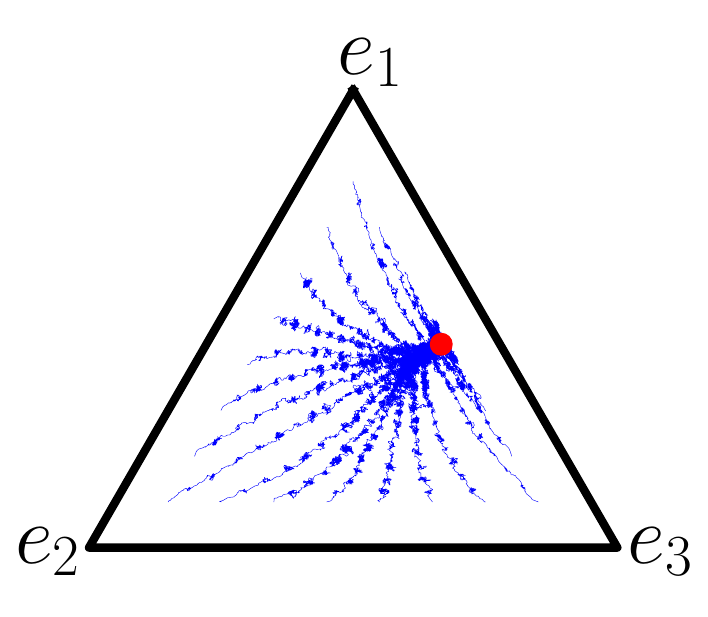}
    \label{fig:kld_rl_finite_population_a}
  }
  \subfigure[]{
    \includegraphics[trim={.0in .2in 0 .1in},width=1.6in]{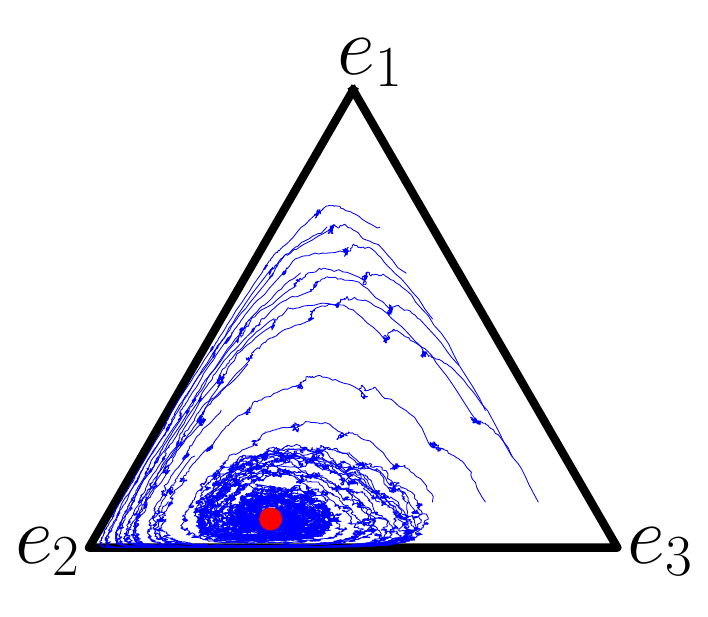}
    \label{fig:kld_rl_finite_population_b}
  }
    \caption{State trajectories of population~1, consisting of a finite number of $1000$ agents, derived using an i.i.d. Poisson process with parameter~$1$ and the KLD-RL protocol \eqref{eq:kld_rl}. The payoff vector is determined by (a) the payoff function with a fixed unit time delay, with $\mathcal F$ given by \eqref{eq:congestion_game} and (b) the smoothing PDM with $\mathcal F$ defined by \eqref{eq:rps_game}. The parameter $\eta$ of \eqref{eq:kld_rl} is selected as $\eta=4.5$ for (a) and $\eta=0.6$ for (b). The red circle in (a) and (b) marks the Nash equilibrium of each game.}
    \label{fig:kld_rl_finite_population}
\end{figure}

  \subsection{Distributed Parameter Update in Finite Populations}
  To evaluate Algorithm~\ref{algorithm:distributed_parameter_update}, we conduct simulations with the same two-population congestion games and zero-sum games in a finite population setting, where each population~$k$ consists of $N^k= 1000$ agents. Unlike the infinite population scenarios in \S\ref{sec:simulation_a}, where the closed-loop model \eqref{eq:closed_loop} was used for the simulations, we consider the scenario where individual agents use an i.i.d. Poisson process with parameter~$1$ and the KLD-RL protocol~\eqref{eq:kld_rl} for the strategy revision. The regularization weight~$\eta$ for \eqref{eq:kld_rl} is determined as in \S\ref{sec:simulation_a}. 
  
  The payoff vector $p^k(t)$ for each population~$k$ is determined by \eqref{eq:population_game_time_delay} with a fixed unit time delay, i.e., $d(t) = 1,~\forall t \geq 0$, for the congestion game, and by \eqref{eq:smoothing_pdm} with parameter $\lambda = 1$ for the zero-sum game. Neighbors of each agent are defined by a strongly connected graph generated by the Erd{\H o}s-R{\' e}nyi model, with an edge formation probability $p=0.1$, and communication between agents takes place at intervals of $T_{\text{comm}} = 0.1$. Fig.~\ref{fig:kld_rl_finite_population} illustrates the outcome of the simulations. As there are a finite number of agents, each element of the population state~$x^k(t)$ takes a value in a discrete space $\{0, 1/N^k, 2/N^k, \cdots, 1 \}$. Consequently, the population state trajectories tremble as they approach the Nash equilibrium, indicated by the red circle, in comparison with Fig.~\ref{fig:kld_rl}.

\section{Conclusions} \label{sec:conclusion}
We studied a multi-agent decision problem in population games in which decision-making agents repeatedly revise their strategy choices based on time-delayed payoffs. We illustrated that under the existing logit model, the agents' strategy revision process may oscillate or converge to the perturbed Nash equilibrium. To address such limitation of existing models, we proposed the KLD-RL model and rigorously prove that the new model allows the agents to asymptotically attain the Nash equilibrium despite the time delays.

As future directions, we plan to apply the KLD-RL model in relevant engineering applications such as multi-vehicle route planning, multi-robot task allocation, and security in power systems, where time delays are inherent in underlying dynamic processes and communication systems.

\section*{Acknowledgement}
The authors would like to thank the anonymous reviewers for their constructive comments that led to substantial improvements in the paper.

\appendix

\change{
\subsection{Definition of Causal Mapping} \label{def:causal_mapping}
We adopt the definition of the \textit{causal mapping} from \cite[Definition~1.1.3]{10.5555/3133758} to represent the class of population games considered in this paper. Specifically, we denote by $\mathfrak G$ a mapping from a population state trajectory to a payoff vector trajectory, capturing the causal relationship between the payoff vector and the current as well as past population states.

Let $\mathcal X_e$ denote the set of $\mathbb X$-valued, time-dependent functions $x: \mathbb R_{\geq 0} \to \mathbb X$ that are square integrable over finite intervals: $\mathcal X_e = \{ x: \mathbb R_{\geq 0} \to \mathbb X \,|\, \int_0^T \|x(t) \|_2^2 \,\mathrm dt < \infty, \, \forall T > 0 \}$. Similarly, let $\mathcal P_e$ represent the set of $\mathbb R^n$-valued, time-dependent functions $p: \mathbb R_{\geq 0} \to \mathbb R^n$ that are square integrable over finite intervals.

To formalize that the payoff vector at any time $t$ depends only on the current and past population states, we define the truncations of $x \in \mathcal X_e$ and $p \in \mathcal P_e$. For $x \in \mathcal X_e$, its truncation to the interval $[0, T]$ is given by: $x_{T} (t) = \begin{cases} x(t) & \text{if } t \in [0, T) \\ 0 & \text{otherwise} \end{cases}$. Similarly, $p_{T}$ represents the truncation of $p \in \mathcal P_e$ to the interval $[0,T]$. Formally, a mapping $\mathfrak G: \mathcal X_e \to \mathcal P_e$ is called \textit{causal} if, for every $x \in \mathcal X_e$ and $T >0$, the following holds: given $p  = \mathfrak G (x)$ and $\bar p = \mathfrak G (x_{T})$, it must be true that $p_{T} (t) = \bar p_{T} (t), ~ \forall t \geq 0$.}

\subsection{Proofs of Lemmas~\ref{lemma:antipassivity_delayed_pdm} and \ref{lemma:antipassivity_smoothing_pdm}} \label{proof_lemma:antipassivity}

  \subsubsection{Proof of Lemma~\ref{lemma:antipassivity_delayed_pdm}}
  By Assumption~\ref{assumption:payoff_function}, it holds that
  \begin{align} \label{eq:boundedPayoffFunction}
    \tilde z' D \mathcal F' (z) D \mathcal F (z) \tilde z \leq B_{D \mathcal F}^2 \tilde z' \tilde z, ~ \forall z \in \mathbb X, \tilde z \in T\mathbb{X},
  \end{align}
  where $D \mathcal F$ is the differential of $\mathcal F$ and $B_{D \mathcal F}$ is a non-negative constant. Using \eqref{eq:boundedPayoffFunction}, we can derive the following relations:
  \begin{align}
    &\int_{t_0}^t \dot{x}'(\tau) \dot{p}(\tau) \,\mathrm d\tau \nonumber \\
    &=\int_{t_0}^t (1 - \dot d(\tau)) \dot{x}'(\tau)D 
      \mathcal F(x(\tau-d(\tau))) \dot{x}(\tau-d(\tau)) \,\mathrm d\tau \nonumber \\
    &\overset{(i)}{\leq} \frac{B_{D \mathcal F}}{2} \int_{t_0}^t \dot{x}'(\tau) \dot{x}(\tau) \,\mathrm d\tau \nonumber \\
    &\quad + \frac{1}{2 B_{D \mathcal F}} \int_{t_0}^t (1 - \dot d(\tau))^2 \, \dot{x}'(\tau-d(\tau))D \mathcal F'(x(\tau-d(\tau))) \nonumber \\
    & \qquad \qquad \qquad \qquad \qquad \times D \mathcal  F (x(\tau-d(\tau))) \, \dot{x}(\tau-d(\tau)) \,\mathrm d\tau \nonumber \\
    &\overset{(ii)}{\leq} \frac{B_{D \mathcal F}}{2} \!\int_{t_0}^t \dot{x}'\!(\tau) \dot{x}(\tau) \,\mathrm d\tau \nonumber \\
    &\qquad + \frac{B_{D \mathcal F} (1 + B_{\dot d})}{2} \int_{t_0-B_{d}}^{t} \dot{x}' (s)\dot{x}(s) \,\mathrm ds \nonumber \\
    &\leq \frac{B_{D \mathcal F} (2 + B_{\dot d})}{2} \, \int_{t_0}^t \dot{x}'(\tau) \dot{x}(\tau) \,\mathrm d\tau \nonumber \\
    &\qquad + \frac{B_{D \mathcal F} (1 + B_{\dot d})}{2} \int_{t_0-B_d}^{t_0} \dot{x}'(\tau)\dot{x}(\tau) \,\mathrm d\tau.
  \end{align}
  To show that $(i)$ and $(ii)$ hold, we use \eqref{eq:boundedPayoffFunction}, the change of variable $s = \tau - d(\tau)$, and the following fact:
  \begin{align*}
    &(1-\dot d(\tau)) \dot{x}'(\tau) D \mathcal F (x(\tau-d(\tau))) \dot{x}(\tau-d(\tau)) \\
    &\leq \frac{B_{D \mathcal F} }{2} \dot{x}'(\tau) \dot{x}(\tau) \\
    &\qquad +\frac{(1-\dot d(\tau))^2}{2 B_{D \mathcal F}} \dot{x}'(\tau-d(\tau)) D \mathcal F'(x(\tau-d(\tau))) \\
    &\qquad\qquad\qquad\qquad\qquad \times D \mathcal F(x(\tau-d(\tau))) \dot{x}(\tau-d(\tau)).
  \end{align*}

  Therefore, by defining $\alpha_{x,p} (t_0) = \frac{B_{D \mathcal F} (1 + B_{\dot d})}{2} \int_{t_0-B_d}^{t_0} \| \dot{x} (\tau)\|_2^2 \, \mathrm d\tau$, we conclude that  \eqref{eq:population_game_time_delay} is weakly $\delta$-antipassive with deficit $\frac{B_{D \mathcal F} (2 + B_{\dot d})}{2}$. \QED

  \subsubsection{Proof of Lemma~\ref{lemma:antipassivity_smoothing_pdm}}
  We proceed by following similar steps as in
  the proof of \cite[Proposition~7]{Park2019Payoff-Dynamic-}. Using the linear superposition principle, we express the solution of \eqref{eq:smoothing_pdm} as
  \begin{align*}
    p(t) &= \underbrace{\exp(-\lambda (t-t_0)) ( p(t_0) - F x(t_0) - b ) + F x(t_0) + b}_{=p^h(t)} \\
         &\qquad + \underbrace{\lambda \int_{t_0}^t \exp(-\lambda (t - \tau)) F \left( x(\tau) - x(t_0) \right) \,\mathrm d \tau}_{=p^f(t)}.
  \end{align*}
  Note that
  \begin{align} \label{eq:inequality_homogeneous}
    &\int_{t_0}^t \dot x'(\tau) \dot p^h (\tau) \,\mathrm d\tau \nonumber \\
    &= \lambda \int_{t_0}^t \dot x'(\tau) \exp(-\lambda(\tau - t_0)) ( F x(t_0) + b - p(t_0) ) \,\mathrm d\tau \nonumber \\
    &\leq \sqrt{n} \| F x(t_0) + b - p(t_0) \|_2.
  \end{align}
  To derive the inequality, we use the fact that the strategy revision protocol $\mathcal T_{ij}^k (x^k(t), p^k(t))$ in \eqref{eq:edm} is a probability distribution, and it holds that $0 \leq \mathcal T_{ij}^k (x^k(t), p^k(t)) \leq 1$ and $\textstyle\sum_{j=1}^{n^k} \mathcal T_{ij}^k (x^k(t), p^k(t)) = 1$. 
  Consequently, in conjunction with \eqref{eq:edm}, we can derive $|\dot x_i^k(t) | \leq 1$, and hence $\|\dot x(t)\|_2 \leq \sqrt{n}$ for all $t \geq 0$, where $n = \sum_{k=1}^M n^k$.

  Let us define $\alpha_{x,p}(t_0) = \sqrt{n} \left\| F x(t_0) + b - p(t_0) \right\|_2.$
  Now we proceed to show that 
  \begin{align} \label{eq:inequality_forced_part}
    \int_{t_0}^t \dot x'(\tau) \dot p^f (\tau) \,\mathrm d\tau \leq \int_{t_0}^t \bar \nu \dot x' (\tau) \dot x (\tau) \,\mathrm d\tau
  \end{align}
  with $\bar \nu$ given in the statement of Lemma~\ref{lemma:antipassivity_smoothing_pdm}.
  By noting that the transfer function from $\dot x$ to $\dot p^f$ is $\frac{\lambda}{\lambda + s} F$ and using Parseval's theorem, \eqref{eq:inequality_forced_part} holds if the following inequality is satisfied:
  \begin{align} \label{eq:condition_in_frequency_domain}
    z^* \!\left( \frac{\lambda}{\lambda \!+\!  j \omega} F \!+\! \frac{\lambda}{\lambda \!-\! j \omega} F' \right) z \!\leq\! 2 \bar \nu z^* z, \,\forall \omega \!\in\! \mathbb R, z \!\in\! T \mathbb C^n,
  \end{align}
  where $T \mathbb{C}^n$ is the tangent space of the $n$-dimensional complex space $\mathbb{C}^n$.
  Note that \eqref{eq:condition_in_frequency_domain} can be re-written as
  \begin{align} \label{eq:inequality_freq_domain}
    \lambda^2 z^* \!\left( F \!+\! F' \right)\! z \!-\! j \lambda \omega z^* \!\left( F \!-\! F' \right)\! z
    \leq 2 \bar \nu \!\left( \lambda^2 \!+\! \omega^2 \right)\! z^* z.
  \end{align}
  By representing $F$ as the sum of its symmetric and skew-symmetric parts, $F = F_{\tiny \text{sym}} + F_{\tiny \text{skew}}$, and the complex variable $z$ as $z = z_R + j z_I$ with $z_R, z_I \in \mathbb R^n$, we can rewrite \eqref{eq:inequality_freq_domain} into
  \begin{multline} \label{eq:inequality_freq_domain_02}
    \lambda^2 \left( z_R' F_{\text{sym}} z_R + z_I' F_\text{sym} z_I \right) - 2 \lambda \omega z_I' F_\text{skew} z_R \\
    \leq \bar \nu \left( \lambda^2 + \omega^2 \right) (z_R' z_R + z_I' z_I).
  \end{multline}

  Since $F$ is contractive, $F_{\text{sym}}$ is a negative-semidefinite matrix in $T \mathbb X$ and for the inequality \eqref{eq:inequality_freq_domain_02} to hold, it suffices to show that
  \begin{align} \label{eq:lemma:antipassivity_smoothing_pdm_01}
    - \frac{2 \lambda \omega}{\bar \nu \left( \lambda^2 + \omega^2 \right)} z_I' F_\text{skew} z_R  \leq z_R' z_R + z_I' z_I.
  \end{align}
  Note that $z_R' z_R + z_I' z_I + \frac{2 \lambda \omega}{\bar \nu ( \lambda^2 + \omega^2 )} z_I' F_\text{skew} z_R \geq z_I' ( I - \frac{1}{4{\bar \nu}^2} F_\text{skew}' F_\text{skew} ) z_I$
  Hence, if $\bar \nu$ satisfies $\bar \nu \geq \frac{1}{2} \| F_\text{skew} \|_2$ then the inequality \eqref{eq:lemma:antipassivity_smoothing_pdm_01} holds.
  In conjunction with \eqref{eq:inequality_homogeneous}, we conclude that
  \begin{align}
    \int_{t_0}^t \left( \dot x'(\tau) \dot p (\tau) - \nu \dot x'(\tau) \dot x(\tau) \right) \,\mathrm d\tau \leq  \alpha_{x, p} (t_0)
  \end{align}
  holds for any $\nu \geq \frac{1}{4} \| F - F' \|_2$ and $\alpha_{x,p} (t_0) = \sqrt{n} \| F x(t_0) + b - p(t_0) \|_2$. \QED

\subsection{Proof of Lemma~\ref{lemma:stability}} \label{proof_lemma:stability}
Recall that $\{t_l\}_{l=1}^\infty$ is the sequence of time instants at which the parameter $\theta$ of the KLD-RL EDM is updated and $\{\theta_l\}_{l=1}^\infty$ is the sequence of resulting parameters, i.e., $\theta_l = x(t_l), ~\forall l \in \mathbb N$, satisfying \eqref{eq:parameter_update_conditions}. We provide a two-part proof: In the first part, we establish that $\{\theta_l\}_{l=1}^\infty$ converges to $\mathbb{NE}(\mathcal F)$. Then, using the result from the first part, we show that the social state $x(t)$ converges to $\mathbb{NE}(\mathcal F)$.

\paragraph{Part I}
Let $x^{\text{\scriptsize NE}}$ be an arbitrary Nash equilibrium in  $\mathbb{NE} (\mathcal F)$. By Definition~\ref{def:nash}, it holds that
\begin{align} \label{eq:XAstNash}
  & ( x^{\text{\scriptsize NE}} - z )' \mathcal F ( x^{\text{\scriptsize NE}} ) \geq 0, ~ \forall z \in \mathbb X.
\end{align}
Note that \eqref{eq:parameter_update_conditions_a} suggests that $\{\theta_l\}_{l=1}^\infty$ satisfies
\begin{multline} \label{eq:InequalityResult01}
  ( \theta_{l+1} - x^{\text{\scriptsize NE}} )' \left( \mathcal F ( \theta_{l+1} ) - \eta \nabla \mathcal D ( \theta_{l+1} \,\|\, \theta_l ) \right) \\ \geq - \frac{\eta}{2} \mathcal D ( \theta_{l+1} \,\|\, \theta_l ).
\end{multline}
Since $\mathcal F$ is a contractive population game under both conditions \ref{stability_contractive_games}) or \ref{stability_strictly_contractive_games}) stated in the lemma, using \eqref{eq:contractive} and \eqref{eq:XAstNash}, it holds that
\begin{multline} \label{eq:InequalityResult02}
  ( \theta_{l+1} \!-\! x^{\text{\scriptsize NE}} )' \mathcal F ( \theta_{l+1} ) \!\leq\! ( \theta_{l+1} \!-\! x^{\text{\scriptsize NE}} )'  \mathcal F ( x^{\text{\scriptsize NE}} ) \!\leq\! 0.
\end{multline}
Also, using \eqref{eq:kl_divergence} and \eqref{eq:gradient_kld}, we can establish
\begin{align} \label{eq:kld_at_nash_equilibrium}
  ( x^{\text{\scriptsize NE}} )' \nabla \mathcal D ( \theta_{l+1} \| \theta_l ) = \mathcal D ( x^{\text{\scriptsize NE}} \| \theta_l ) - \mathcal D ( x^{\text{\scriptsize NE}} \| \theta_{l+1} ).
\end{align}
Using \eqref{eq:InequalityResult01}-\eqref{eq:kld_at_nash_equilibrium}, we can derive
\begin{align*}
  0 &\geq ( \theta_{l+1} - x^{\text{\scriptsize NE}} )' \nabla \mathcal D ( \theta_{l+1} \,\|\, \theta_l ) - \frac{1}{2} \mathcal D ( \theta_{l+1} \| \theta_l ) \nonumber \\
    &= \mathcal D ( \theta_{l+1} \| \theta_l ) \!-\! \mathcal D ( x^{\text{\scriptsize NE}} \| \theta_l ) \!+\! \mathcal D ( x^{\text{\scriptsize NE}} \| \theta_{l+1} ) \!-\! \frac{1}{2} \mathcal D ( \theta_{l+1} \| \theta_l )
\end{align*}
and hence we conclude
\begin{align} \label{eq:InequalityResult04}
  \mathcal D ( x^{\text{\scriptsize NE}} \| \theta_{l+1} ) \leq \mathcal D ( x^{\text{\scriptsize NE}} \| \theta_l ) - \frac{1}{2} \mathcal D ( \theta_{l+1} \| \theta_l ).
\end{align}

From \eqref{eq:InequalityResult04}, we can infer that $\{ \mathcal D ( x^{\text{\scriptsize NE}} \| \theta_l ) \}_{l=1}^\infty$ is a monotonically decreasing sequence, and since every $\mathcal D ( x^{\text{\scriptsize NE}} \| \theta_l)$ is nonnegative, it holds that 
$\lim_{l \to \infty} \mathcal D ( \theta_{l+1} \| \theta_l ) = 0$.
Therefore, according to \eqref{eq:parameter_update_conditions_a}, we have that 
\begin{align} \label{eq:convergence_condition}
  \lim_{l \to \infty} \max_{z \in \mathbb X} ( z - \theta_{l+1} )' ( \mathcal F(\theta_{l+1}) - \eta \nabla \mathcal D ( \theta_{l+1} \| \theta_l ) ) = 0.
\end{align}
Let $\theta_l = (\theta_{l}^1, \cdots, \theta_{l}^{M}) \in \mathbb X$ and $\theta_l^k = (\theta_{l,1}^k, \cdots, \theta_{l,n^k}^k) \in \mathbb X^k$. Since $\mathcal D ( \theta_{l+1} \| \theta_l )$ converges to zero as $l$ tends to infinity, $\lim_{l \to \infty} \theta_{l,i}^k = 0$ implies $\lim_{l \to \infty} \theta_{l+1,i}^k = 0$. Also if $\theta_l$ converges to $\bar \theta$, then so does $\theta_{l+1}$.

To complete the proof of Part~I, it is sufficient to show that when a subsequence $\{\theta_{l_m}\}_{m=1}^\infty$ converges to $\bar \theta \in \mathbb X$,
it holds that $\bar \theta \in \mathbb{NE}(\mathcal F)$. From \eqref{eq:convergence_condition}, we can derive
\begin{align} \label{eq:inequality_over_support_subset}
  &\max_{z \in \mathbb X_{\bar \theta}} ( z - \bar \theta )' \mathcal F(\bar \theta) \nonumber \\
  &= \!\lim_{m \to \infty} \max_{z \in \mathbb X_{\bar \theta}} ( z \!-\! \theta_{l_m\!+\!1} )' \!( \mathcal F(\theta_{l_m\!+\!1}) \!-\! \eta \nabla \mathcal D ( \theta_{l_m\!+\!1} \| \theta_{l_m} ) ) \nonumber \\
  &\leq \!\lim_{m \to \infty} \max_{z \in \mathbb X} ( z \!-\! \theta_{l_m\!+\!1} )' \!( \mathcal F(\theta_{l_m\!+\!1}) \!-\! \eta \nabla \mathcal D ( \theta_{l_m\!+\!1} \| \theta_{l_m} ) ) \!=\! 0,
\end{align}
where $\mathbb X_{\bar \theta}$ is a subset of $\mathbb X$ given by $\mathbb X_{\bar \theta} = \{x \in \mathbb X \,|\, x_i^k = 0 \text{ if } \bar \theta_i^k = 0\}$.
Note that every Nash equilibrium $x^{\text{NE}}$ of $\mathcal F$ belongs to $\mathbb X_{\bar \theta}$; otherwise there is $x^{\text{NE}} \in \mathbb {NE}(\mathcal F)$ for which $\lim_{m \to \infty} \mathcal D(x^{\text{NE}} \| \theta_{l_m}) = \infty$, which contradicts the fact that $\{\mathcal D(x^{\text{NE}} \| \theta_{l})\}_{l=1}^\infty$ is a decreasing sequence. Then, the following inequalities hold:
\begin{subequations} \label{eq:inequalities_for_limit_points}
  \begin{align} 
    ( x^{\text{NE}} - \bar \theta )' \mathcal F(\bar \theta) \leq 0 \\
    ( x^{\text{NE}} - \bar \theta )' \mathcal F(x^{\text{NE}}) \geq 0.
  \end{align}
\end{subequations}

If the condition \ref{stability_contractive_games}) of the lemma holds, then $\mathbb X_{\bar \theta} = \mathbb X$ and, hence, by \eqref{eq:inequality_over_support_subset}, $\bar \theta$ is a Nash equilibrium. In addition, if the condition \ref{stability_strictly_contractive_games}) holds, then using \eqref{eq:inequalities_for_limit_points} and since $\mathcal F$ is strictly contractive, we can establish $( x^{\text{NE}} - \bar \theta )' (\mathcal F(x^{\text{NE}}) - \mathcal F(\bar \theta)) = 0$ which implies $\bar \theta = x^{\text{NE}}$. Hence, every limit point $\bar \theta$ of $\{\theta_l\}_{l=1}^\infty$ belongs to $\mathbb{NE} (\mathcal F)$. Therefore, we conclude that the sequence $\left\{ \theta_l \right\}_{l=1}^\infty$ converges to the Nash equilibrium set $\mathbb{NE} (\mathcal F)$.

\paragraph{Part II}

Let $\alpha_{x,p}$ be the upper bound of the inequality \eqref{eq:weakAntipassivity}. By \eqref{eq:parameter_update_conditions_b}, it holds that $\lim_{l \to \infty} \alpha_{x,p}(t_l) = 0$. Also, let $\mathcal S_{\theta_l}(x(t), p(t))$ be the informative $\delta$-storage function of the KLD-RL EDM, defined as in \eqref{eq:DeltaStorageFunc}, over the time interval $t \in [t_l, t_{l+1})$. Using \eqref{eq:weakAntipassivity} and \eqref{eq:delta_passivity}, we can derive
\begin{equation} \label{eq:delta_storage_function_sequence}
  \mathcal S_{\theta_l} ( x(t), p(t) ) \!\leq\! \mathcal S_{\theta_l} ( x(t_l), p(t_l) ) \!+\! \alpha_{x,p} (t_l), ~ \forall t \!\in\! [t_l, t_{l+1}).
\end{equation}
Given that $\theta_l = x(t_l)$, the term $\mathcal S_{\theta_l} ( x(t_l), p(t_l) )$ can be rewritten as
\begin{align}
  &\mathcal S_{\theta_l} ( x(t_l), p(t_l) ) \nonumber \\
  &= \max_{\bar z \in \mathrm{int}(\mathbb X)} (\bar z' p(t_l)  - \eta \mathcal D ( \bar z \,\|\, \theta_l ) ) - \theta_l' p(t_l) \nonumber \\
  &= \eta \textstyle \sum_{k=1}^M \ln \left( \textstyle \sum_{s=1}^{n^k} \theta_{l,s}^k \exp(\eta^{-1} p_s^k(t_l)) \right) - \theta_l' p(t_l),
\end{align}
where we use \eqref{eq:maximization} and \eqref{eq:kld_rl} to derive the last equality. By \eqref{eq:parameter_update_conditions_b} and \eqref{eq:inequality_over_support_subset}, we observe that for any convergent subsequence $\{\theta_{l_m}\}_{m=1}^\infty$ of $\{\theta_{l}\}_{l=1}^\infty$ with its limit point $\bar \theta$ in $\mathbb{NE}(\mathcal F)$, it holds that $\lim_{m \to \infty} \ln ( \textstyle \sum_{s=1}^{n^k} \theta_{l_m,s}^k \exp(\eta^{-1} p_s^k(t_{l_m})) ) = \eta^{-1} \max_{j \in \{1 \leq s \leq n^k \,|\, \bar \theta_s^k > 0\}} \mathcal F_j^k(\bar \theta)$,
from which we can establish that
\begin{align} \label{eq:convergence_S_theta_l}
  \lim_{l \to \infty} \mathcal S_{\theta_l} ( x(t_l), p(t_l) ) = 0.
\end{align}
Consequently, for any $\epsilon > 0$, we can find a positive integer $L$ for which
\begin{align} \label{eq:storage_function_convergence}
  \mathcal S_{\theta_l} ( x(t), p(t) ) \leq \epsilon, ~ \forall t \in [t_l, t_{l+1})
\end{align}
holds for all $l \geq L$.

Note that $\mathcal D(z \,\|\, \theta_l)$ is a strongly convex function of $z$, i.e., $\nabla^2 \mathcal D(z \,\|\, \theta_l) \geq I$. According to the analysis used in \cite[Theorem~2.1]{10.2307/3081987}, 
\begin{align*}
  \eta \tilde z' \nabla \mathcal D(y \,\|\, \theta_l) = \tilde z'r \iff y = \argmax_{z \in \mathrm{int} ( \mathbb X )} ( z'r - \eta \mathcal D(z \,\|\, \theta_l) )
\end{align*}
holds for all $r \in \mathbb R^n$, $y \in \mathrm{int}(\mathbb X)$, and $\tilde z \in T\mathbb X$. Let $y(t) = \argmax_{z \in \mathrm{int} ( \mathbb X )} ( z' p(t) - \eta \mathcal D(z \,\|\, \theta_l) )$, then we can derive
\begin{multline} \label{eq:storage_function_strong_convexity}
  \mathcal S_{\theta_l} (x(t), p(t)) = \eta (y(t) - x(t))' \nabla \mathcal D (y(t) \,\|\, \theta_l) \\
  - \eta (\mathcal D (y(t) \,\|\, \theta_l) - \mathcal D (x(t) \,\|\, \theta_l)) \geq \frac{\eta}{2} \|y(t) - x(t)\|_2^2,
\end{multline}
where to establish the inequality, we use the strong convexity of $\mathcal D(z \,\|\, \theta_l)$.
In conjunction with \eqref{eq:storage_function_convergence} and \eqref{eq:storage_function_strong_convexity}, by the definition of the KLD-RL protocol \eqref{eq:maximization}, we conclude that $\lim_{t \to \infty} \|\dot x(t)\|_2 = 0$.

We complete the proof by showing that the social state $x(t)$ converges to $\mathbb{NE} (\mathcal F)$. By contradiction, suppose that there is a sequence $\{\tau_m\}_{m=1}^\infty$ of time instants for which $\left\{ x(\tau_m) \right\}_{m=1}^\infty$ converges to $\bar x$, where $\bar x$ does not belong to $\mathbb{NE} (\mathcal F)$. Let $\{[t_{l_m}, t_{l_m+1})\}_{m=1}^\infty \subseteq \{[t_l, t_{l+1})\}_{l=1}^\infty$ be the collection of time intervals for which $\tau_m \in [t_{l_m}, t_{l_m+1})$ holds for all $m \in \mathbb N$. By Assumption~\ref{assumption:pdm}-1,
\eqref{eq:delta_storage_function_sequence}, and \eqref{eq:convergence_S_theta_l}, we have that
\begin{align} \label{eq:limit_of_sequence}
  \max_{\bar z \in \mathbb X} \left( \bar z' \mathcal F ( \bar x ) \!-\! \eta \mathcal D ( \bar z \| \bar \theta ) \right)
  \!-\! \left( \bar x' \mathcal F(\bar x) \!-\! \eta \mathcal D ( \bar x \| \bar \theta ) \right) \!=\! 0,
\end{align}
where $\bar \theta$ is a limit point of the subsequence $\left\{\theta_{l_m} \right\}_{m=1} ^\infty$ of $\left\{\theta_l \right\}_{l=1} ^\infty$.
By \eqref{eq:limit_of_sequence} and the facts that $\bar \theta$ belongs to $\mathbb{NE} (\mathcal F)$ and $\mathcal F$ is contractive, we can derive
\begin{align*}
  \eta \mathcal D ( \bar x \,\|\, \bar \theta ) \leq (\bar x - \bar \theta)' \mathcal F(\bar x) \leq (\bar x - \bar \theta)' \mathcal F(\bar \theta) \leq 0, 
\end{align*}
which yields that $\bar x = \bar \theta$.
This contradicts our hypothesis that $\bar x$ does not belong to $\mathbb{NE} (\mathcal F)$. \QED

\subsection{Proofs of Theorems~\ref{theorem:stability_time_delay} and \ref{theorem:stability_smoothing_pdm}} \label{proof_theorem:stability}

\subsubsection{Proof of Theorem~\ref{theorem:stability_time_delay}}
According to Lemmas~\ref{lemma:antipassivity_delayed_pdm}-\ref{lemma:stability}, it suffices to show that \eqref{eq:bias_selection_criteria_1} of Algorithm~\ref{algorithm:parameter_update} implies \eqref{eq:parameter_update_conditions}.
Let $\{\theta_l\}_{l=1}^\infty$ be a sequence of parameter updates determined by \eqref{eq:bias_selection_criteria_1} of Algorithm~\ref{algorithm:parameter_update}. Note that
\begin{align} \label{eq:parameter_update_conditions_a_time_delay_case}
  &\max_{z \in \mathbb X} ( z - \theta_{l+1} )' ( \mathcal F(\theta_{l+1})
    - \eta \nabla \mathcal D ( \theta_{l+1} \| \theta_l ) ) \nonumber \\
  &\leq \max_{z \in \mathbb X} ( z - \theta_{l+1} )' ( p(t_{l+1})) - \eta \nabla \mathcal D ( \theta_{l+1} \| \theta_l ) )  \nonumber \\
  &\qquad + \sqrt{2M} \| p(t_{l+1}) - \mathcal F(\theta_{l+1}) \|_2 \,.
\end{align}
We can establish a bound on $\| p(t_{l+1}) - \mathcal F(\theta_{l+1}) \|_2$ as follows:
\begin{align} \label{eq:bound_on_payoff_difference}
  &\| p(t_{l+1}) - \mathcal F(\theta_{l+1}) \|_2 \nonumber \\
  &= \| \mathcal F(x(t_{l+1} - d(t_{l+1}))) - \mathcal F(x(t_{l+1})) \|_2 \nonumber \\
  &\leq \max_{z \in \mathbb X} \| D \mathcal F (z) \|_2 \| x(t_{l+1} - d(t_{l+1})) - x(t_{l+1}) \|_2 \nonumber \\
  &\leq B_{D\mathcal F} B_d \max_{\tau \in [t_{l+1} - B_d, t_{l+1}]} \|\dot x (\tau)\|_2
\end{align}
where $B_d$ upper bounds $d(t)$ in \eqref{eq:population_game_time_delay} and $B_{D\mathcal F}$   upper bounds  $D \mathcal F$.
Hence, by \eqref{eq:parameter_update_conditions_a_time_delay_case} and \eqref{eq:bound_on_payoff_difference}, \eqref{eq:bias_selection_criteria_1} implies \eqref{eq:parameter_update_conditions_a}. To establish \eqref{eq:parameter_update_conditions_b}, suppose $\lim_{l \to \infty} \mathcal D (\theta_{l+1} \| \theta_l) = 0$. Using \eqref{eq:bias_selection_criteria_1}, we can derive $\lim_{l \to \infty} \max_{\tau \in [t_{l+1} - B_d, t_{l+1}]} \| \dot x(\tau) \|_2 = 0$ from which we conclude that
\begin{align*}
  &\lim_{l \to \infty} \| p(t_l) - \mathcal F(x(t_l)) \|_2 = 0 \\
  &\lim_{l \to \infty} \alpha_{x, p} (t_l) = \lim_{l \to \infty} \frac{B_{D\mathcal F} (1 + B_{\dot d})}{2} \int_{t_l-B_d}^{t_l} \left\| \dot{x} (\tau) \right\|_2 \, \mathrm d\tau = 0.
\end{align*}
This completes the proof. \QED

\subsubsection{Proof of Theorem~\ref{theorem:stability_smoothing_pdm}}
Similar to the proof of Theorem~\ref{theorem:stability_time_delay}, we show that \eqref{eq:bias_selection_criteria_2} of Algorithm~\ref{algorithm:parameter_update} implies \eqref{eq:parameter_update_conditions}.
We begin with establishing a bound on $\| p(t_{l+1}) - \mathcal F(x(t_{l+1})) \|_2$. Consider the following differential equation:
\begin{align*}
  \frac{\mathrm d}{\mathrm dt} ( p(t) - \mathcal F ( x(t) ) ) = - \lambda ( p(t) - \mathcal F ( x(t) ) ) - D \mathcal F ( x(t) ) \dot x(t).
\end{align*}
For any $t_0 < t_{l+1}$, a solution to the equation at $t = t_{l+1}$ can be computed as 
\begin{multline*}
  p(t_{l+1}) - \mathcal F ( x(t_{l+1}) ) 
  \!=\! ( p (t_{0}) - \mathcal F ( x(t_{0}) ) ) \exp (-\lambda (t_{l+1} - t_{0} )) \\
  - \int_{t_{0}}^{t_{l+1}} \exp({-\lambda (t_{l+1} - \tau )}) D \mathcal F ( x(\tau) ) \dot x(\tau) \, \mathrm d\tau.
\end{multline*}
By assigning $t_0 = \gamma t_{l+1}$ with $\gamma \in (0,1)$, we obtain
\begin{align} \label{eq:bound_on_payoff_difference_smoothing_pdm}
  &\| p(t_{l+1}) - \mathcal F ( x(t_{l+1}) ) \|_2 \nonumber \\
  &\leq ( \| p(\gamma t_{l+1}) \|_2 + B_{\mathcal F} ) \exp(-\lambda (1 - \gamma ) t_{l+1}) \nonumber \\
  &\qquad + B_{D \mathcal F} \int_{\gamma t_{l+1}}^{t_{l+1}} \exp({-\lambda (t_{l+1} - \tau)}) \| \dot x(\tau) \|_2 \, \mathrm d\tau.
\end{align}
Hence, in conjunction with \eqref{eq:parameter_update_conditions_a_time_delay_case} and \eqref{eq:bound_on_payoff_difference_smoothing_pdm}, we observe that
\eqref{eq:bias_selection_criteria_2} implies \eqref{eq:parameter_update_conditions_a}. To establish \eqref{eq:parameter_update_conditions_b}, using \eqref{eq:bias_selection_criteria_2} and \eqref{eq:bound_on_payoff_difference_smoothing_pdm}, if $\lim_{l \to \infty} \mathcal D (\theta_{l+1} \| \theta_l) = 0$ then it holds that $\lim_{l \to \infty} \left\| p(t_l) - \mathcal F \left( x(t_l) \right) \right\|_2 = 0,$
which implies $\lim_{l \to \infty} \alpha_{x,p} (t_l) = \lim_{l \to \infty} \sqrt{n} \left\| p(t_l) - F x(t_l) - b \right\|_2 = 0,$
where we use $\mathcal F(x) = Fx + b$.
This completes the proof.
\QED

\balance
\bibliographystyle{IEEEtran}
\bibliography{IEEEabrv,references}

\end{document}